\newcounter{propertyCounter}

\newcommand{\authorA}{Marco~B. Caminati}
\newcommand{\Title}{A strengthening of the MCFL-ness of $O_2$}
\newcommand{\letters}{\Sigma}
\newcommand{\A}{\conj a}
\newcommand{\B}{\conj b}
\newcommand{\C}{\conj c}
\newcommand{\proj}{\pi}

\newcommand{\counterexampleP}{\mathsf p}
\newcommand{\ceP}{\counterexampleP}
\newcommand{\counterexampleQ}{\mathsf q}
\newcommand{\ceQ}{\counterexampleQ}

\newcommand{\N}{\mathbb{N}}
\newcommand{\Z}{\mathbb{Z}}
\newcommand{\ve}[1]{\mathbf{#1}}
\newcommand{\x}{\ve x}

\newcommand{\ze}{\emptyset}
\newcommand{\bumps}{B}
\newcommand{\Left}{[}
\newcommand{\Right}{]}
\newcommand{\Up}{\sqcap}
\newcommand{\Down}{\sqcup}
\newcommand{\basis}{\mathbf{e}}
\newcommand{\eq}{\equiv}

\newcommand{\climple}{P}
\newcommand{\candidatempty}{Z}

\newcommand{\reversed}[1]{\underline{#1}}
\newcommand{\rev}{\reversed}
\newcommand{\trapezoidals}{T}
\newcommand{\selffact}{R}
\newcommand{\clearbumps}{C}
\newcommand{\cloggedbumps}{F}
\newcommand{\mcfg}{MCFG}
\newcommand{\criticalbumps}{G}

\newcommand{\size}{\card}
\newcommand{\norm}[1]{\lVert #1 \rVert}
\newcommand{\restrict}[2]{\left. #1 \right|_{#2}}

\newcommand{\sdiff}{\backslash{}}
\newcommand{\ol}[1]{\overline{#1}}
\newcommand{\Conj}{\iota}
\newcommand{\conj}{\ol}
\newcommand{\conc}{*}
\newcommand{\doublebumps}{D}
\newcommand{\enum}{\tau}
\newcommand{\cornerbumps}{E}

\newcommand{\iso}[1]{\widetilde{#1}}
\newcommand{\ortho}{\bot}

\newcommand{\subsequence}[2]{ {\left \langle {#1} \right \rangle }_{#2}}

\newcommand{\factor}[3]{ {\left \langle {#1} \right \rangle }_{#2}^{#3}}

\newcommand{\proves}[1]{\sststile{#1}{}}
\newcommand{\ra}[1]{\xLongrightarrow{\text{\tiny{{#1}}}}}

\documentclass
{article}

\usepackage{amsmath}
\usepackage{amsthm}
\usepackage{amsthm}
\usepackage{amssymb} 
\usepackage[inline]{enumitem}
\usepackage{extarrows}
\usepackage[T1]{fontenc}
\usepackage{geometry}
\usepackage{graphicx}
\usepackage{hyperref}
\usepackage{listings}
\usepackage{microtype}
\usepackage{turnstile}
\usepackage{wrapfig}

\hypersetup{%
pdftitle={\Title
}%
, pdfauthor={\authorA}%
, pdfstartview=FitBH
, pdffitwindow=true
, pdfkeywords={logic, proof theory, formalization, completeness, model theory, Mizar, satisfiability theorem, \L{}-Skolem theorem, automated proof checking}
, pdflang=en-US
, pdfstartview= 100 100 10000
, pdfsubject={}
, pdfwindowui=true
, colorlinks=true
, linkcolor=blue
, urlcolor=blue
, citecolor=red
, breaklinks = true
, menucolor=brown
}

\DeclareMathOperator{\dom}{dom}
\DeclareMathOperator{\ran}{ran}
\DeclareMathOperator{\card}{card}
\DeclareMathOperator{\argmin}{argmin}

\theoremstyle{plain}
\newtheorem{Cor}{Corollary}[]
\newtheorem{Lm}
{Lemma}
\newtheorem{Prop}
{Proposition}[]
\newtheorem{Thm}
{Theorem}
\theoremstyle{definition}
\newtheorem{Def}
{Definition}
\newtheorem{Rem}
{Remark}
\newtheorem{Not}
{Notation}
\newtheorem{Ex}
{Example}

\begin{document}
\title{\Title}
\author{\authorA{}
\thanks{Lancaster University, Leipzig, Germany}\\
\href{mailto:m.caminati@lancaster.ac.uk}{m.caminati@lancaster.ac.uk} \\
\texttt{ORCID: \href{https://orcid.org/0000-0002-4529-5442}{0000-0002-4529-5442}}
}


\maketitle

\begin{abstract}
In the last years, a number of proofs of the fact that $O_2$ is a multiple context-free grammar (\mcfg) were given. Such results can be exploited in the fields of both computational linguistics and of computational algebra. Here, we focus on a recent such proof spelled in terms of factorizations of string tuples, and give a new result with a stronger characterization of such factorizations than in existing theorems.
\end{abstract}


\section{Introduction}
For a given natural $N$, we consider $N$ letters and, for each of them, another letter called its inverse, thus obtaining a set $\letters_N$ of $2N$ distinct letters; $O_N$ is the sub-language of $\letters_N^*$ consisting of exactly all the words where the number of occurrences of any letter equals the number of occurrences of the inverse of that letter.
From a computational linguistics perspective, it is noteworthy to determine where the $O_N$ family of languages fits within the established, rich hierarchy of grammars designed to model natural languages. 
From the standpoint of computational group theory, insights from formal language theory can aid in investigating related groups~\cite{gilman2005formal,gilman2018groups}. This connection has motivated a series of studies~\cite{salvati2015mix,nederhof2016short}, culminating in a result showing that $O_N$ can be generated by an $N$-multiple context-free grammar ($N$-\mcfg) for any value of $N$~\cite{gebhardt2022n}.
The \mcfg{} definition (see for example~\cite{clark1985introduction,caminati2024Dlt}) provides little control on the sentential forms introduced at each step of a derivation, for example on their lengths.
However, the fact that $O_2$ is a multiple context-free language (MCFL) remains valid even if we replace the MCFG given in~\cite{caminati2024Dlt} with a stronger form of grammar which is no longer a MCFG, giving more precise information on how the parse tree of a string in $O_2$ looks like.
This is expressed formally by the following theorem, which is the main result of this paper; more specifically, the ``more precise information'' is given by the side conditions of rule~\eqref{RefRuleRectangular}:

\begin{Thm}
\label{RefLmMain}
Given a pair of words $\left( p_0, p_1 \right)$, assume $p_0 p_1 \in O_2$. 
Then $ \left( p_0, p_1 \right)$ can be derived using the following rules:
\begin{small}
\begin{align}
\label{RefRule0}
\left( p_0, \conj {p_0} \right) & \ \to \ && \text{where } 
p_0 \in \letters_2^* \text{ and }
\left| p_0 \right| \le 1
\\
\label{RefRuleRectangular}
\left( p_0 p_1, q_0 q_1 \right) &\ \to \ \left( p_0, q_0 \right), \left( p_1, q_1 \right) 
&&
\text{where } 
\left| p_0 \right| = \left| q_0 \right| = 1 \text{ or }
\left| p_1 \right| = \left| q_1 \right| = 1
\\
\label{RefRuleSandwich}
\left( p_0 p_1, q_1 q_0 \right) & \ \to \  \left( p_0, q_0 \right),  \left( p_1, q_1 \right)
\\
\label{RefRuleLeft}
\left( p_0 p_1 q_0, q_1 \right) & \ \to \  \left( p_0, q_0 \right),  \left( p_1, q_1 \right)
\\
\label{RefRuleRight}
\left( p_0, p_1 q_0 q_1 \right) & \ \to \  \left( p_0, q_0 \right),  \left( p_1, q_1 \right)
\end{align}
\end{small}
\end{Thm}

In the derivation rules of Theorem~\ref{RefLmMain}, 
$p_0$, $q_0$, $p_1$, $q_1$ are variables representing words.
Rule \eqref{RefRule0} has an empty right hand, meaning that the corresponding pair on the left hand can be derived without any previous operation: such rules are also called \emph{axioms}.
The other rules can be read as follows: 
if in previous steps both pairs on the right hand side of the arrow have been derived, then the pair on the left hand side can also be derived in the current step.
Iterating this procedure any finite number of times,  pairs $\left( p,q \right)$ can be derived for some particular words $p, q$ of $\letters_2^*$; if a given pair $ \left( p, q \right)$ is thus derivable using rules \eqref{RefRule0}-\eqref{RefRuleRight}, we will write $ \proves{} \left( p, q \right)$.
The task is to show that any pair $\left( p,q \right)$ such that $pq \in O_2$ can be derived.

In~\cite{caminati2024Dlt}, the main effort was to provide a short and elementary proof amenable to proof assistants.
In stark contrast, the apparently small addition of the condition in rule~\eqref{RefRuleRectangular} notably complicates the proof, requiring an analysis and categorisation of the occurrence of bumps already introduced in~\cite{caminati2024Dlt}.
After Section~\ref{RefSectNotation}, where notations are introduced, Section~\ref{RefSectBump} introduces the notion of bump, thanks to which the main result is reduced to two theorems, respectively proved in Sections~\ref{RefSectLmMain1} and~\ref{RefSectLmMain2}. Section~\ref{RefSectConcl} concludes.

\section{Notations}
\label{RefSectNotation}
$\N{}$, $\Z^+$, $\Z$ are, respectively, the set of natural numbers (including $0$), the set of the positive integers and of the integers.
$\cup$, $\cap$, $\sdiff$ and $\times$ are binary, infix symbols for the set-theoretical operations of union, intersection, set difference and cartesian product, respectively;
$\dom$ and $\ran$ are the operators returning the domain and range (also called image), respectively, of a given relation or function (i.e., $\ran Q = Q \left( \dom Q \right)$ for any relation or function $Q$).
$Q^{-1}$ is the inverse of the relation or function $Q$, while $\restrict Q X$ is its restriction to the set $X \subseteq \dom Q$, and $g \circ f$ is the composition of the two functions $g$ and $f$.
Throughout the paper, we will use square brackets to denote intervals of integers:
$
\left[ i, j \right] := \left\{ k \in \Z. i \le k \le j \right\}$
and
$
\left] i, j \right[ := \left[ i, j \right] \sdiff \left\{ i, j \right\}
\forall i, j \in \Z.
$
A \emph{string} or \emph{word} is a function with domain $\left[ 0, j \right]$ for some $j \in \Z$.
A function will be called a \emph{sequence} if it is either a string or has $\N$ as domain.
Therefore, if $p$ is a non-empty sequence, $p\left( 0 \right)$ is its first entry.
$\card Y$ is the cardinality of the set $Y$; when $Y$ is a sequence, $\card Y$ represents its length;%
\footnote{This is consistent with the fact that, under the standard representation of functions as sets of cartesian pairs, the cardinality and length of a sequence do coincide.}
we will often write $\left| Y \right|$ as a shorthand for $\card Y$.
The \emph{reverse} of a string $p$, denoted $\rev p$, is the unique string $q$ such that $ \left| p \right| = \left| q \right|$ and $q \left( i \right) = p \left( \left| p \right| - 1 - i \right) \forall i \in \dom p$; in particular, if $p \neq \ze$, the last entry of $p$ can be indicated with $\rev p \left( 0 \right)$, a notation that we will use often in this paper.
Given a set $X$, $X^*$, $X^{\infty}$ and $X^N$ are, respectively, the set of the strings over $X$ (i.e., whose range is included in $X$), of the sequences over $X$, and of the sequences over $X$ having length $N$ for some $N\in \N$.
Given  $X \subseteq \N$ and a sequence $q$, $\subsequence{q}{X}$ is the subsequence of $q$ obtained by cancelling all its entries whose indices are not in $X$.
More formally, build the strictly monotonic sequence $p$ of the natural numbers belonging to $X \cap{} \dom q$, and set $\subsequence{q}{X} := q \circ p$.
We will employ the special notation $p - X := \subsequence{p}{\N \sdiff X}$. 
When $p = \subsequence{q}{X}$ and $X$ is an integer interval $[i, j]$, $p$ is called a \emph{factor} of $q$ (a \emph{left factor} if $i=0$, a \emph{right factor} if $q$ is a string and $j=\left| q \right| - 1$), and we will just write $\factor{q}{i}{j}$ instead of $\subsequence{q}{\left[ i, j \right]}$;
a \emph{proper factor} of $q$ is a factor of $q$ not equal to $q$.
Given a word $p_0$ and a sequence $p_1$, the concatenation of $p_0$ and $p_1$ (written $p_0 \conc p_1$, or just $p_0 p_1$ when safe) is the unique sequence $p$ such that $\subsequence{p}{\left[ 0, \left| p_0 \right| - 1  \right]} = p_0$ and 
$p - \left[ 0, \left| p_0 \right| - 1 \right] = p_1$.

We will be studying strings over a set $\letters$ (our \emph{alphabet}) possessing a certain structure, which we now introduce.
Let $\enum$ be a given, fixed injection defined on $\Z \sdiff \left\{ 0 \right\}$
such that $\letters := \ran \enum$ is disjoint from $\Z$, and consider the injective involution on $\letters$ given by
$\Conj := \letters \ni x \mapsto \enum \left( - \enum^{-1} \left( x  \right)  \right) \in \letters $.
An element of $\letters$ will also be called a \emph{letter}; a letter $x$ is said to \emph{occur} in $p \in \letters^*$ if $x \in \ran p$, and the \emph{occurrences} of $x$ in $p$ are the indices of the set $p^{-1} \left( \left\{ x \right\}  \right)$.
We will usually write $\conj x$ in lieu of $\Conj \left( x \right)$, and 
will say that $\conj x$ is the \emph{inverse} of $x$.
We denote with $\letters^{(+)}$ the set $\enum \left( \Z^+ \right)$, with $\letters_N$ the set 
$ \left\{ x \in \letters: \left| \enum^{(-1)} \left( x \right)  \right| \le N \right\}$, 
and with $\letters_N^{(+)}$ the set $\letters^{(+)} \cap \letters_N$, where $N$ is any natural number.
Furthermore, with a slight abuse, we will also write $\conj p$, with $p \in \letters^*$, to mean $\Conj \circ p$ (i.e., the string obtained by inverting all the letters of $p$).
Given $x, y \in \letters$, we write $x // y$ to mean $y \in \left\{ x, \conj x \right\}$, 
$x \ortho y$ (and say that $x$ and $y$ are orthogonal) to mean $\neg \left( x // y \right)$, and $x^{\ortho}$ for the set of all letters in $\letters$ orthogonal to $x$.
$\enum$ naturally induces a linear well-ordering on $\letters^{(+)}$: hence, we can use the first lowercase characters $a, b, c, d, \ldots$ to denote its first elements so that, for example, 
$\letters_3 = \left\{ a, \A, b, \B, c, \C  \right\}$ and
$\letters_2 = \left\{ a, \A, b, \B  \right\}$.
A particular string $p$ can be fully described by writing down either the tuple
$ \left( p\left( 0 \right), p \left( 1 \right), \ldots, p \left( \left| p \right| - 1 \right)\right) $ or the concatenation of all its entries (where parentheses can be omitted thanks to associativity of concatenation): 
$ p\left( 0 \right) p \left( 1 \right) \ldots p\left( \left| p \right| - 1 \right) $.
We will generally adopt the first writing when the string belongs to $\Z^N$ for some $N\in \N$, and the second when the string is in $\letters^*$.

The converse of Theorem~\eqref{RefLmMain} is proven by an easy induction argument (omitted):
\begin{Prop}
\label{RefLmConverse}
If $\proves{} \left( p_0, q_0 \right)$, 
then $p_0 q_0 \in O_2$. 
\end{Prop}



The proof of Theorem~\ref{RefLmMain} is by contradiction: if it is false,  
then we are allowed to pick a pair $ \left( \counterexampleP, \counterexampleQ \right)$ 
being not derivable using rules \eqref{RefRule0}-\eqref{RefRuleRight}, satisfying 
$\counterexampleP \counterexampleQ  \in O_2 $, and minimal in the sense that any other $ \left( p, q \right)$ satisfying these two properties must necessarily also satisfy $ \left| p q \right| \ge \left| \counterexampleP \counterexampleQ \right|$.
A quick check shows that it must be 
$\left| \counterexampleP \counterexampleQ \right| \ge 4$.
Using rules \eqref{RefRuleLeft} and \eqref{RefRuleRight}, we deduce that it also must be 
$\left| \counterexampleP  \right| \neq 1 \neq \left|  \counterexampleQ \right|$.
Again using rules \eqref{RefRuleLeft} and \eqref{RefRuleRight}, we can assume 
$\left| \counterexampleP  \right| \neq 0 \neq \left|  \counterexampleQ \right|$ because, for example, if a pair $\left( \ze, q \right)$ satisfies the same requirements and is not derivable, then any proper binary factorisation $\left( q_0, q_1 \right)$ of $q$ is also not derivable (if it were, $\left( \ze, q_0 \ze q_1 = q \right)$ would also be derivable using rule~\eqref{RefRuleRight}).
To recapitulate, if Theorem~\ref{RefLmMain} is false, then we can find $\left( \counterexampleP, \counterexampleQ \right) \in \letters_2^* \times \letters_2^*$ such that:
\begin{enumerate}[label=P\arabic*]
\item
\label{RefReqBalanced}
: $ \counterexampleP \counterexampleQ \in O_2$;
\item
\label{RefReqNoReach}
: $\not \proves{} \left( \counterexampleP, \counterexampleQ \right)$;
\item
\label{RefReqMin}
: Any other pair $ \left( p, q \right) \in \letters_2^* \times \letters_2^*$ satisfying 
\ref{RefReqBalanced} and \ref{RefReqNoReach}
must also satisfy $\left| p q \right| \ge \left| \counterexampleP \counterexampleQ \right|$;
\item
\label{RefReqSize}
: $\min \left\{  \left| \counterexampleP \right|, \left| \counterexampleQ  \right| \right\} \ge 2$;
\item
\label{RefReqRectangular}
: $\counterexampleP \left( 0 \right) \neq \conj{ \counterexampleQ \left( 0 \right)}$ and 
$\rev{\ceP} \left( 0 \right) \neq \conj{ \rev {\counterexampleQ} \left( 0 \right)}$,
\setcounter{propertyCounter}{\value{enumi}}
\end{enumerate}

where~\ref{RefReqRectangular} imposes that the first letters of $\ceP$ and $\ceQ$, respectively, must not be related by inversion, and that the same holds for their last letters.
This follows immediately from \ref{RefReqBalanced}-\ref{RefReqSize} by exploiting the particular form of rule \eqref{RefRuleRectangular}.
It is clear that to proceed we must use the remaining rules (\eqref{RefRuleSandwich}, \eqref{RefRuleLeft} and \eqref{RefRuleRight}) to draw more properties of $ \left( \ceP, \ceQ \right)$.
This is better done by introducing a new representation of $\letters^*$, which is the task of the next section.
There,
we will also use the new representation to 
state further properties of $ \left( \counterexampleP, \counterexampleQ \right)$, and show that they together imply that $ \left( \counterexampleP, \counterexampleQ \right)$ cannot exist.
We will reach this conclusion by stating and using two theorems, \ref{RefLmMain1} and \ref{RefLmMain2}, which are subsequently proven in Sections~\ref{RefSectLmMain1} and \ref{RefSectLmMain2}, respectively.

\section{Representing strings in $\letters_N^*$ as tuples in $\Z^N$}
\label{RefSectBump}
Until the beginning of Section~\ref{RefSectProofSkeleton}, $N$ will be a generic positive integer.
Consider the additive group $ \left( \Z^\infty, + \right) $ of infinite sequences of integers, and let $\basis_i$, for any $i \in \Z^+$, be the element of $\Z^\infty$ having the $i - 1$-th entry set to $1$ and all the remaining ones set to $0$ (e.g., $\basis_1 = \left( 1, 0, 0, \ldots \right)$).
We consider the unique map $\mu$ defined on $\letters$ which associates the $i$-th letter of $\letters^{(+)}$ to $\basis_i$, and preserves inversion: $\mu \left( \conj x \right) = - \mu \left( x \right)$.
This map can be naturally extended to a monoid morphism between $ \left( \letters^*, \conc \right)$ (where $\conc$ denotes word concatenation) and $ \left( \Z^\infty, + \right) $ by imposing $\mu \left( p q \right) = \mu \left( p \right) + \mu \left( q \right)$; 
moreover, the restriction of $\mu$ to any $\letters_N^*$ can be identified with a monoid morphism from $ \left( \letters_N^*, \conc \right)$ onto $ \left( \Z^N, + \right)$, which we  indicate with $\mu_N$, or just $\mu$ when safe.
On $\Z^N$, we also define $N$ maps $ \proj_i, i = 1, \ldots, N$, each taking a tuple and returning its $i$-th component (e.g., $\proj_2 \left( \left( 0,4,3 \right)  \right) = 4$).
Intuitively, each letter of a word $p \in  \letters_N^* $ can be regarded as a unit step along a certain axis of $\Z^N$, where the step is in the positive verse if the letter is in $\letters^{\left( +  \right)}$ and in the opposite direction otherwise; $\mu \left( p \right)$ yields the vector locating the position in $\Z^N$ obtained after performing, starting from the origin, the steps thus dictated by all the letters of $p$.
We therefore call $\mu \left( p \right)$ the \emph{displacement} of $p$; 
for example, $\mu \left( a b b \conj a \right) = \mu \left( a \right) + \mu \left( b \right) + \mu\left( b \right) - \mu\left( a \right) = \left( 1,0 \right) + \left( 0,1 \right) + \left( 0,1 \right) - \left( 1, 0 \right) = \left( 0, 2 \right)$.

Two words $p$, $q$ such that for any letter the difference between the occurrences of that letter and the occurrences of its inverse is the same for $p$ and for $q$ will return the same value through $\mu$: we will write $p \eq q$. 
Formally, $\eq$ is the equivalence relation canonically associated to $\mu$: 
$ p \eq q \Leftrightarrow \mu \left( p \right) = \mu \left( q \right)$; 
$\eq_N$ will denote its restriction to $\letters_N^*$, but we will just write $\eq$ when this yields no ambiguity.
$ \left[ p \right]_{\eq}$ denotes the equivalence class to which $p$ belongs.
Obviously, given $p \in \letters_N^*$, $p \in O_N$ if and only if $\mu \left( p \right) = 0_N$; in this case, we will say that $p$ is \emph{closed}.
A \emph{loop} for $p \in \letters^*$ is a proper factor of $p$ which belongs to
$ \left[ \ze \right]_{\eq} \sdiff \left\{ \ze \right\}$; if $p$ has no loops, it is \emph{simple}; 
$p$ is \emph{self-factoring} if it admits a proper factor equivalent to $p$.
Note that any non-empty closed string is self-factoring.

While length is defined for any string, a string $p \in \letters_N^*$ can also be associated another natural number thanks to the representation just described.
Recall that the Manhattan (or taxicab) distance between two $N$-tuples $\left( x_1, \ldots, x_n \right)$, $ \left( y_1, \ldots, y_N  \right)$ of integers is defined as 
$ \sum_{j=1}^N \left| y_j - x_j \right|$; we denote with $ \norm{\ve x}$ the Manhattan distance of $ \ve x \in \Z^N$ from the origin, and we will often use the shorthand notation $\norm p$ to mean $\norm{\mu_N \left( p \right)}$ for a given $p \in \letters_N^*$, calling $\norm p$ the \emph{radius} of $p$.
Hence, we now have two natural numbers associated to any $p \in \letters_N^*$: $\left| p \right|$ is the number of letters of $p$, while $\norm p$ is the minimal number of unit steps along all the $N$ directions needed to arrive to $\mu \left( p \right)$ from the origin.
It is obvious that $\norm p = \min_{\card} \left[ p \right]_{\eq}$ for any $p \in \letters_N^*$: that is, among all strings equivalent to $p$, those with fewest letters have exactly $\norm p$ letters; such strings will be called \emph{short}.
An integer interval $\left[ i, j \right]$ is a \emph{detour} for $p \in \letters_N^*$ if 
$ \left[ i, j \right] \subseteq \dom p$ and $\factor p i j $ is not short (note that this implies $j > i$ and $\left| p \right| \ge 2$); 
in this case, $\factor p i j$ will also be called a detour for $p$, and its radius will be referred to as the radius of the detour.
Any detour of a string can be replaced with something shorter without changing the displacement: more precisely, $ \left( i, j, p \right)$ is an $m$-\emph{shortcut} for $q$ if $ \left[ i, j \right]$ is a detour for $q$, 
$m = \left| p \right| < \left| \factor {q} {i} {j} \right|$, and 
$ p \eq_N \factor q i j$; 
$ \norm{p}$ will be called the radius of the shortcut, and 
the string $ \factor {q} {0} {i-1} \ p \ \factor {q} {j+1} {\left| q \right| - 1}$ will also be called an $m$-shortcut for $q$.
A detour $\left[ i,j \right]$ for $p$ is an $x$-\emph{bump} for $p$ if no proper subset of $ \left[ i, j \right]$ is a detour for $p$ (i.e., any proper factor of $ \factor p i j$ is short) and $ p \left( i \right) = x $; 
note this implies that $ p\left( j \right) = \conj x$, and that 
$ p \left( k \right) \notin \left\{ x, \conj x \right\} \forall k \in \left] i, j \right[$.
$x$ is called the \emph{direction} of the bump.
$\bumps^x \left( p \right)$ is the set of all $x$-bumps of $p$, and we set 
$ \bumps \left( p \right) := \bigcup_{x \in \letters_N} \bumps^x \left( p \right)$.
If $\left[ i, j \right]$ is a detour (bump) for $p$, 
the string $ \factor p i j $
will also be called a detour (bump) for $p$.
$p$ is short if and only if it has no detours, and any detour contains a bump. 
Hence:
\begin{Rem}
$p$ is short if and only if $\bumps \left( p \right) = \ze$.
\end{Rem}

Note that if $\left[ i, j \right]$ is a detour for $q$, it is always possible to find $p$ such that $\left( i, j, p \right) $ is a shortcut for $q$.
In the particular case of $\left[ i, j \right]$ being a bump, there is a particularly simple way of performing this operation:
\begin{Rem}
\label{RefLmBumpCancel}
If $\left[ i, j \right]$ is a bump for $p$, then
$p - \left\{ i, j \right\} \eq p$, and $ p - \left\{ i, j \right\}$ is a shortcut for $p$.
\end{Rem}

Recall that $ p -\left\{ i, j \right\}$ denotes the operation of cancelling the letters of indices $i$ and $j$ from the string $p$.
This operation will be extensively used in this paper, especially in the case of $ \left[ i, j \right]$ being a bump for $p$ (\emph{bump cancelling}).
In particular, we will typically be dealing with a pair $ \left( p_0, p_1 \right)$ of strings, and wanting to perform bump cancelling on exactly one between $p_0$ and $p_1$.
To find a convenient notation in this situation, we first need to introduce a way of denoting all the bumps of $p_0$ and $p_1$ together, but with the possibility of distinguishing which bumps belongs to which string.
A notation to attain this goal is the following:
\begin{Def}
\begin{align*}
\bumps^x \left( p_0, p_1 \right) := 
\bumps^x \left( p_0 \right) \cup \left\{ \left[ i + \left| p_0 \right|, j + \left| p_0 \right| \right]. \ \left[ i, j \right] \in \bumps^x \left( p_1 \right) \right\},
&& 
\bumps \left( p_0, p_1 \right) := \bigcup_{x \in \letters_N} \bumps^x \left( p_0, p_1 \right)
\end{align*}
\end{Def}

Now, if we define 
\begin{Def}
$ \left( p_0, p_1 \right) - X := \left( p_0 - X, 
\factor {p_0 p_1 - X}{\left| p_0 \right|}{\left|p_0 p_1  \right| - 1} 
\right)$,
\end{Def}
then, when writing $ \left( p_0, p_1 \right) - \left\{ i, j \right\}$ for some $ \left[ i, j \right] \in \bumps \left( p_0, p_1 \right)$, it will be clear which bump we are cancelling from which string.
Given $x \in \letters_N$, by an $x$-bump for $ \left( p_0, p_1 \right)$ we mean any element of $ \bumps^x \left( p_0, p_1 \right)$, while by a \emph{minimal} bump for $ \left( p_0, p_1 \right) $ we mean an element of $\argmin_{\card} \bumps^y \left( p_0, p_1 \right)$ for some $y \in \letters_N$.
Note that $\bumps^x \left( p_0, p_1 \right) \subseteq \bumps^x \left( p_0 p_1 \right)$,
and that the converse inclusion is not always true.
We will say that $\left( q_0, q_1 \right)$ is an $m$-shortcut for $ \left( p_0, p_1 \right)$ if there is $k \in \left\{ 0, 1 \right\}$ such that 
$q_k$ is an $m$-shortcut for $p_k$ and $q_{\left| k - 1 \right|} = p_{\left| k - 1 \right|}$.

\begin{Ex}
\label{RefExampleFirst}
Let $p_0 := a b a a \conj b \conj b $, 
$p_1 := \conj a \conj a \conj b \conj a b \conj a \conj b \conj a \conj a b \conj a b a a a a$ 
$\in \letters_2^*$.
Then 
$ \left( p_0, p_1  \right) - \left\{ 8, 10 \right\}$ is the $1$-shortcut of $ \left( p_0, p_1 \right)$ given by $ \left( p_0, p_1' \right)$, where
$p_1' := \conj a \conj a \conj a \conj a \conj b \conj a \conj a b \conj a b a a a a$.

$ \bumps^b \left( p_0, p_1 \right) = \left\{ \left[ 1, 4 \right], \left[ 10, 12 \right] \right\}$, 
and the minimal bumps for $ \left( p_0, p_1 \right)$ are exactly $3$: 
$ \left[ 8, 10 \right]$,
$ \left[ 10, 12 \right]$, and
$ \left[ 16, 18 \right]$; note that $ \left[ 3, 6 \right] \in 
\bumps^a \left( p_0 p_1 \right) \sdiff 
\bumps^a \left( p_0, p_1 \right)
$ is a bump for $ p_0 p_1$ but not for $ \left( p_0, p_1 \right)$.
Another $1$-shortcut of $ \left( p_0, p_1 \right)$ is given by
$\left(  p_0,
\factor {p_1} 0 4
\ 
b \right) =
\left( p_0, p_1 \right) -
\left( \left[ 11, 16 \right] \cup
\left[ 
18, \left| p_0 p_1 \right| - 1
 \right]
\right)
$.
The first shortcut is obtained by applying Remark~\ref{RefLmBumpCancel}, while the second is of more general nature, being obtained by substituting a detour of $p_1$ (having radius $1$) which is not a bump.
\end{Ex}

Equivalent strings (i.e., strings for which the value of $\mu_N$ is the same) will result in the same displacement in different ways, and the definition of $\mu_N$ does not 
include the exact points of $\Z^N$ that the string touches to arrive to its displacement; 
we now introduce a notation to represent this additional information.

\begin{Def}
Given $p \in \letters_N^*$, define 
$\iso {p} := \left[ 0, \left| p \right| \right] \ni k \mapsto \mu_N \left( 
\factor p 0 {k-1}
\right).
$
\end{Def}

$\iso p \left( 0 \right)$ is the origin, $ \iso p \left( 1 \right)$ is the displacement after stepping according to the first letter of $p$, and so on; 
and, naturally, 
$ \iso p \left( \left| p \right| \right) = \mu \left( p \right)$.
An $N$-dimensional \emph{walk} is any string $q \in \left( \Z^N  \right)^*$ such that
$q = \factor {\iso p} {i} j$ for some $p \in \letters_N^*$, $i, j \in \N$.
Note that $\iso p$ is a sequence of sequences in $\Z^N$, and that not every
sequence of sequences in $\Z^N$ is a walk.
Referring back to Example~\ref{RefExampleFirst}, two possible bidimensional walks we can construct are $\iso {p_0} = \left(  
\left( 0,0 \right),
\left( 1,0 \right),
\left( 1,1 \right),
\left( 2,1 \right),
\left( 3,1 \right),
\left( 3,0 \right),
\left( 3,-1 \right)
\right) $
and
$ \factor{\iso {p_0 p_1}} 5 7=
\left( 
\left( 3, 0 \right),
\left( 3, -1 \right),
\left( 2, -1 \right)
\right).
$
To familiarise with the concepts introduced, we present a result which will be useful in the sequel.
It shows that two bidimensional short walks joining pairs of suitably positioned points must meet.

\begin{Lm}
\label{RefLmGeometric}
Given $p_0$, $p_1$, $q_0$, $q_1 \in \letters_2^*$, assume:
\begin{enumerate*}
\item
\label{RefAssFirst}
\label{RefAssPositionP}
$\left\{ \iso {p_0} \left( \left| p_0 \right| \right),
\iso {p_0 p_1} \left( \left| p_0 p_1 \right| \right) \right\} 
= \left\{ \left( x_0, y_0 \right), \left( x_2, y_1 \right) \right\};$
\item
\label{RefAssPositionQ}
$\left\{ \iso {q_0} \left( \left| q_0 \right| \right),
\iso {q_0 q_1} \left( \left| q_0 q_1 \right| \right) \right\} 
= \left\{ \left( x_1, y_0 \right), \left( x_3, y_1 \right) \right\};$
\item
\label{RefAssCross}
$x_1 \geq x_0$ and $x_2 \geq x_3;$
\item
\label{RefAssShort}
\label{RefAssLast}
$ \bumps \left( p_1 \right) = \bumps \left( q_1 \right) = \ze{}.$
\end{enumerate*}
Then
$ \ran \left( \iso {p_0 p_1}|_{\left[ \left| p_0 \right|, \left| p_0 p_1 \right| \right]} \right) 
\cap
\ran \left( \iso {q_0 q_1}|_{\left[ \left| q_0 \right|, \left| q_0 q_1 \right| \right]} \right) \neq \ze{}.
$
\end{Lm}

Hypothesis~\eqref{RefAssPositionP} requires that the walk whose steps are given by $p_1$ and whose starting point is given by $\mu_2 \left( p_0 \right)$ goes from 
$ \left( x_0, y_0 \right)$ to $ \left( x_2, y_1 \right)$ 
or vice versa.
Similarly, hypothesis~\eqref{RefAssPositionQ} requires the same for $q_0$, $q_1$, 
$ \left( x_1, y_0 \right)$, $ \left( x_3, y_1 \right)$.
Hypothesis~\eqref{RefAssCross} imposes a condition on the relative horizontal position of the endpoints of the walk described by hypothesis~\eqref{RefAssPositionP} with respect to the horizontal position of the endpoints of the walk described by hypothesis~\eqref{RefAssPositionQ}.
Finally, hypothesis~\eqref{RefAssShort} requires that both the walks are short.
Lemma~\ref{RefLmGeometric} then states that the two walks must intersect: 
indeed, 
$
\ran \left( \iso {p_0 p_1}|_{\left[ \left| p_0 \right|, \left| p_0 p_1 \right| \right]} \right) 
$
is the set of points of $\Z^2$ touched by the first walk.
\begin{figure}
\includegraphics[trim={24mm 5mm 4.7cm 6cm}, clip, scale=1.1]{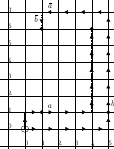}
\hfill
\includegraphics[trim={2mm 36mm 70mm 29mm}, clip, scale=1.1]{f00}
\hfill
\includegraphics[trim={33mm 36mm 37mm 32mm}, clip, scale=1.1]{f00}
\caption{Three applications of Lemma~\ref{RefLmGeometric} to the case 
$x_0 = 1$, $ y_0 = 1$, $x_2=4$, $y_1=6$, $x_1=5$, $x_3=1$.}
\label{RefFigCross}
\end{figure}%
Before giving the proof, we loiter on possible applications of Lemma~\ref{RefLmGeometric}
in the configuration of points represented in Figure~\ref{RefFigCross}, where 
the circled point locates the origin, 
$x_0 = 1$, $ y_0 = 1$, $x_2=4$, $y_1=6$, $x_1=5$, $x_3=1$, the dashed line represents
the walk given by $p_1$, and the dotted line that given by $q_1$;
the endpoints of the walks are marked with solid dots.
In the first case (Figure~\ref{RefFigCross}, left), the two walks do not intersect because, while $p_1$ is short, $q_1$ is not, being $\bumps \left( q_1 \right) =
\left\{ \left[ 5, 10 \right] \right\}.
$
In the second case (Figure~\ref{RefFigCross}, middle), they do not intersect because, while $q_1$ is short, $p_1$ is not, being $\bumps \left( p_1 \right) =
\left\{\left[ 0, 2 \right], \left[ 3, 9 \right] \left[ 8, 13 \right] \right\}.$
In the third case (Figure~\ref{RefFigCross}, right), they do intersect as mandated by the lemma, because both $p_1$ and $q_1$ are short.
Note that the only role of $p_0$ and $q_0$ in the lemma is to set the displacement; that is, the application of the lemma depends only on $\mu \left( p_0 \right)$ and $\mu \left( q_0 \right)$ rather than on $p_0$ and $q_0$ themselves.

\begin{proof}
Consider the map $f$ associating to each $ \left( y_0, y_1 \right) \in \Z^2$ 
the set consisting of all the tuples\\
$
\left( p_0, p_1, q_0, q_1, x_0, x_1, x_2, x_3 \right)
$
whose entries, together with $y_0$ and $y_1$, satisfy hypotheses \eqref{RefAssFirst}--\eqref{RefAssLast}, together with
$
\ran \left( \restrict {\iso { p_0 p_1}}
{\left[ \left| p_0 \right|, \left| p_0 p_1 \right| \right]} \right) 
\cap
\ran \left( \restrict {\iso { q_0 q_1}}
{\left[ \left| q_0 \right|, \left| q_0 q_1 \right| \right]} \right) 
= \ze$.
If, by contradiction, the set $M$ made of the pairs  $  \left( y_0, y_1 \right)$
for which $f$ returns a non-empty set were non-empty, then we could take
$ \left( \ol {y_0}, \ol{y_1} \right) \in M$ such that $ \left| \ol{y_1} - \ol{y_0} \right|$ is minimal.
It is easy to see that $\ol{y_1} \neq \ol{y_0}$, so that we can assume 
$\ol{y_1} > \ol{y_0}$.
$f \left( \left( \ol{y_0}, \ol{y_1} \right) \right)$ being non empty, we consider 
$ \left( p_0', p_1', q_0', q_1', \ol{x_0}, \ol{x_1}, \ol{x_2}, \ol{x_3} \right)$ in it such that
$ \left| p_1' q_1' \right|$ is minimal.
Set $p := p_0' p_1'$ and $q := q_0' q_1'$; by possibly applying the reversal operation, we can assume 
$ \mu \left( p_0' \right) = \left( \ol{x_0}, \ol{y_0} \right)$,
$ \mu \left( p \right) = \left( \ol{x_2}, \ol{y_1} \right)$,
$ \mu \left( q_0' \right) = \left( \ol{x_1}, \ol{y_0} \right)$,
$ \mu \left( q \right) = \left( \ol{x_3}, \ol{y_1} \right)$.
Further, set 
$i := \max \left( {\iso p }^{-1} \left( \Z \times \left\{ y_1 - 1 \right\} \right) \right)$
(i.e., $i$ is the last step at which the walk $\iso p$ reaches a point of ordinate $\ol{y_1} - 1$), 
$j := \max \left( {\iso q }^{-1} \left( \Z \times \left\{ y_1 - 1 \right\} \right) \right)$, 
$p_1'' := \factor {p} {\left| p_0' \right|} {i-1}$,
$p_1''' := \factor {p} i {\left| p \right|} $,
$q_1'' := \factor {q} {\left| q_0' \right|} {j-1}$,
$q_1''' := \factor {q} j {\left| q \right|}$.
Now, it cannot be $ p_1' \left( 0 \right) = b = q_1' \left( 0 \right)$, otherwise the minimality of $\left| \ol{y_1} - \ol{y_0} \right| $ would be violated.
As a consequence, $\max \left\{ i - \left| p_0' \right|, j - \left| q_0' \right|  \right\} > 0$, and therefore
\begin{align}
\label{RefFormula22}
\left|  p_1''' q_1''' \right| < \left| p_1' q_1' \right|.
\end{align}

We also note that
\begin{enumerate*}
\item
$
\ran \left( \restrict {\iso p} 
{\left[ \left| p_0'  \right|, \left| p_0' p_1''  \right| \right]}
 \right)
\cap 
\ran {\restrict {\iso q}
{\left[ \left| q_0'  \right|, \left|  q_0' q_1''  \right| \right]}
} 
= \ze$ 
because\\
$
\ran \left( \restrict {\iso p}
{\left[ \left| p_0' \right|, \left| p \right| \right]}
\right)
\cap
\ran \left( 
\restrict {\iso q}
{\left[ \left| q_0' \right|, \left| q \right| \right]}
\right) 
= \ze.
$
\item
$\bumps \left( p_1'' \right) = \ze = \bumps \left( q_1'' \right)$
because $\bumps \left( p_1' \right) = \ze = \bumps \left( q_1' \right)$ 
and $p_1''$, $q_1''$ are factors, respectively, of $p_1'$ and of $q_1'$.
\item
$
\left( \proj_1 \circ \mu \right) \left( p_0' p_1'' \right) 
<
\left( \proj_1 \circ \mu \right) \left( q_0' q_1'' \right) 
$
by minimality of $\left| \ol{y_1} - \ol{y_0} \right| $ and the two previous points.
\item
$\bumps \left( p_1''' \right) = \ze = \bumps \left( q_1''' \right)$
because $\bumps \left( p_1' \right) = \ze = \bumps \left( q_1' \right)$
and $p_1'''$, $q_1'''$ are factors, respectively, of $p_1'$ and of $q_1'$.
\end{enumerate*}
If $\ol{y_1} > \ol{y_0} + 1$, then we can use the minimality of $ \left| \ol{y_1} - \ol{y_0} \right|$ and the last two points above to draw that
$
\ran \left( \restrict {\iso {p}} {\left[ \left| p_0' p_1'' \right|, \left| p \right| \right]} \right)
\cap
\ran \left( \restrict {\iso {q}} {\left[ \left| q_0' q_1'' \right|, \left| q \right| \right]} \right) \neq \ze.
$
If $\ol{y_1} = \ol{y_0} + 1$, then we can use~\eqref{RefFormula22} and the last two points in the list above to draw the same conclusion, due to the minimality of $ \left| p_1' q_1' \right|$.
This conclusion clashes with 
$ 
\ran \left( \restrict {\iso p} {\left[ \left| p_0' \right|,\left| p \right|  \right]} \right)
\cap
\ran \left( \restrict {\iso q} {\left[ \left| q_0' \right| , \left| q \right|  \right]} \right)
= \ze{}.
$
\end{proof}

We conclude this section with a final, elementary
\begin{Rem}
\label{RefRem}
Let $p$ be a factor of $q u r \in \letters_N$, with $\left| u \right| \leq 1$.
If $u' \eq u$, then there is a factor $p' \eq p$ of $q u' r$.
\end{Rem}
\begin{proof}
Since $\left| u \right| \leq 1$, $u$ is either fully included in $p$ or not. 
If not, we just take $p' := p$. 
Otherwise, $p'$ is the string obtained by replacing $u$ with $u'$ in $p$.
\end{proof}

\subsection{Proof of Theorem~\ref{RefLmMain}}
\label{RefSectProofSkeleton}
Rules \eqref{RefRule0}--\eqref{RefRuleRight} are invariant with respect to swapping $a$'s with $\conj a$'s, $b$'s with $\conj b$'s, or $a$'s with $b$'s and $\conj a$'s with $\conj b$'s, 
as well as under reversal of both the strings in the pairs involved, and under swapping the words in the pairs involved.
All the notions we will be interested in enjoy the same invariances;
for example, $ \left( p_0, p_1 \right)$ has a $1$-shortcut if and only if $ \left( \rev p_1, \rev p_0 \right)$ has.
We will often use these invariances silently; 
a first consequence is that 
we can impose that both the components of $\mu \left( \ceP \right)$ are non-negative without affecting the other properties we proved up to this point.
We state this new property explicitly in the following proposition, along with others.

\begin{Prop}
\mbox{}
\begin{enumerate*}[label=P\arabic*]
\setcounter{enumi}{\value{propertyCounter}}
\item
\label{RefReqPositive}
If $ \left( m, n \right) = \mu \left( \ceP \right)$, then $\min \left\{ m, n \right\} \ge 0$;
\item
\label{RefReqSimple}
$ \ceP \ceQ $ is closed and simple;
\item
\label{RefReqSelfFact}
$\ceP$ and $\ceQ$ are both non-self-factoring.
\end{enumerate*}
\end{Prop}

\begin{proof}
We discussed \ref{RefReqPositive} above, while closure of $ \ceP \ceQ$ was already established (\ref{RefReqBalanced}).
Assume $ \ceP \ceQ$ has a loop; then the loop must be a factor of either $\ceP$ or $\ceQ$, otherwise we could use rule 
\eqref{RefRuleSandwich} to derive $ \left( \ceP, \ceQ \right)$.
We can suppose then that $\ceP$ has a loop (the other case being symmetric):
$\ceP = p_1 p_0 p_2$, 
$p_0 \in \left[ \ze \right]_{\eq} \sdiff \left\{ \ze, \ceP \right\}$.
We can further strengthen this to $p_1 \neq \ze \neq p_2$ for otherwise we could use rule~\eqref{RefRuleSandwich} to derive $\left( \ceP, \ceQ \right)$.
By \ref{RefReqMin}, $\proves{} \left( p_1 p_2, \ceQ \right)$, hence let us proceed by cases on the last rule applied in the derivation of $\left( p_1 p_2, \ceQ \right) $;
we suppose this rule is~\eqref{RefRuleRectangular}, the other cases being simpler.
Consider $p_3, p_4, q_1, q_2$ such that $p_1 p_2 = p_3 q_1$, $\ceQ = p_4 q_2$, 
$\proves{} \left( p_3, p_4 \right)$, $ \proves{} \left( q_1, q_2 \right)$,
$ \left\{ 1 \right\} \in \left\{ \left\{  \left| p_3 \right|, \left| p_4 \right| \right\}, \left\{
\left| q_1  \right|, \left|  q_2 \right|
\right\} \right\}$.
We only show the subcase $ \left| p_3 \right| = \left| p_4 \right| = 1$, the other being symmetric. 
Therefore, we have $p_1 p_2 = x q_1$ and $\ceQ = \conj x q_2$ for some letter $x \in \letters_2$; since $p_1 \neq \ze$, this implies that $ \ceP = p_1 p_0 p_2 = x q_1' p_0 q_1''$ with $q_1' q_1'' = q_1$.
By~\ref{RefReqMin}, then, $\proves{} \left( q_1' p_0 q_1'', q_2 \right)$, so that we can apply rule~\eqref{RefRuleRectangular} to contradictorily derive $\left( \ceP, \ceQ \right)$.
Hence, $\ceP \ceQ$ cannot have a loop.

It $\ceP$ were self-factoring, then $\ceP = q_0 p q_1$ for some $p$, $q_0$, $q_1$ 
such that $q_0 q_1 \in \left[ \ze \right]_{\eq} \sdiff \left\{ \ze, \ceP \right\}$ (if it were $q_0 q_1 = \ceP$, then $\ceP \ceQ$ would have a loop, which we escluded).
Again by~\ref{RefReqMin}, this implies $\proves{} \left( q_0, q_1 \right)$ and $\proves{} \left( p, \ceQ \right)$, which yields to the contradiction $ \proves{} \left( \ceP, \ceQ \right)$ using rule~\eqref{RefRuleLeft}. 
Symmetrically, using rule~\eqref{RefRuleRight}, $\ceQ$ cannot be self-factoring.
\end{proof}

\ref{RefReqSimple}, \ref{RefReqRectangular} and \ref{RefReqSelfFact} state three properties which $ \left( \ceP, \ceQ \right) $ satisfies; in the sequel, we will be interested in studying whether such properties apply to other string pairs.
Therefore, it is convenient to introduce dedicated definitions to help concisely expressing these properties.

\begin{Def}
\mbox{}
\begin{itemize}
\item
$\climple
:= \left\{ \left( p, q \right) \in \letters_2^* \times \letters_2^* . \
p q \text{ is closed and simple }
\right\} $
\item
$ \trapezoidals 
:= \left\{ \left( p, q \right) \in \letters_2^* \times \letters_2^* . \ 
p \neq \ze \neq q \wedge \left( p \left( 0 \right) = \conj { q \left( 0 \right) } \vee  
\rev p \left( 0 \right) = \conj{ \rev q \left( 0 \right)}
\right)
\right\} $
\item
$
\selffact 
:= \left\{ \left( p, q \right) \in \letters_2^* \times \letters_2^* . \
\text{either } p \text{ or } q \text{ is self-factoring}
\right\} $
\end{itemize}
\end{Def}

\ref{RefReqSimple}, \ref{RefReqRectangular} and \ref{RefReqSelfFact} can thus be condensed into the statement
$
 \left( \ceP, \ceQ \right) \in \climple \sdiff \trapezoidals \sdiff \selffact.
$
We are now ready to prove Theorem~\ref{RefLmMain} using the following two theorems, whose proofs are, respectively, in Sections~\ref{RefSectLmMain1} and \ref{RefSectLmMain2}.

\begin{Thm}
\label{RefLmMain1}
Let $ \left( p_0, p_1 \right) \in \climple \sdiff \trapezoidals \sdiff \selffact$.
Assume
\begin{enumerate*}
\item
\label{RefAss1}
$\bumps \left( p_0, p_1 \right) \neq \ze$, and
\item
\label{RefAss2}
Given any minimal bump $\left[ i, j \right]$ for $ \left( p_0, p_1 \right)$, $ \left( p_0, p_1 \right) - \left\{ i, j \right\} \notin \climple \sdiff \trapezoidals $.
\end{enumerate*}
Then $ \left( p_0, p_1 \right)$ admits a $1$-shortcut which belongs to $\climple \sdiff \trapezoidals$.
\end{Thm}

\begin{Thm}
\label{RefLmMain2}
There is no minimal bump $ \left[ i, j \right]$ for $ \left( \ceP, \ceQ \right)$ such that
$ \left( \ceP, \ceQ \right) - \left\{ i, j \right\} \in \climple \sdiff \trapezoidals.$
\end{Thm}

\begin{proof}[Proof of Theorem~\ref{RefLmMain}]
Theorem~\ref{RefLmMain2} grants that $ \left( \ceP, \ceQ \right)$ satisfies hypothesis~\eqref{RefAss2} of Theorem~\ref{RefLmMain1}; if we also manage to prove 
$ \bumps \left( \ceP, \ceQ \right) \neq \ze$, we can apply Theorem~\ref{RefLmMain1} to 
$ \left( \ceP, \ceQ \right)$, obtaining 
a $1$-shortcut $\left( p, q \right)$ of $ \left( \ceP, \ceQ \right)$, and we can assume $q=\ceQ$.
Then $ \left( p, q \right) \in \selffact$, otherwise $ \left( p, q \right)$ would not be derivable by virtue of~\ref{RefReqMin}.
This means that $p$ is self-factoring and is also a $1$-shortcut of $\ceP$.
Now, it is simple to check that a string admitting a $1$-shortcut which is self-factoring is also self factoring, by just applying Remark~\ref{RefRem}.
This contradicts $ \left( \ceP, \ceQ \right) \notin \selffact$: $ \left( \ceP, \ceQ \right)$ cannot exist, and Theorem~\ref{RefLmMain} holds.

We reduced to prove $ \bumps \left( \ceP, \ceQ \right) \neq \ze$.
Assuming the contrary equates to say that $\ceP$ and $\ceQ$ are both short.
This implies, due to \ref{RefReqPositive}, that $ \left\{ \ceP \left( 0 \right), \rev \ceP \left( 0 \right) \right\} \subseteq \left\{ a, b \right\}$ and 
$
\left\{ \ceQ \left( 0 \right), \rev \ceQ \left( 0 \right) \right\} \subseteq \left\{ 
\conj a, \conj b
\right\}.
$
Properties~\ref{RefReqSimple} and~\ref{RefReqRectangular} further restrict the possibilities to only two: 
$ 
\left(  
\ceP \left( 0 \right), \rev \ceP \left( 0 \right),
\ceQ \left( 0 \right), \rev \ceQ \left( 0 \right)
\right) 
\in 
\left\{ 
\left( a, a, \conj b, \conj b \right),
\left( b, b, \conj a, \conj a \right)
\right\}.
$
In both these configurations, we apply Lemma~\ref{RefLmGeometric} to see that 
there must be distinct $i, j \in \left] 0, \left| \ceP \ceQ \right| - 1 \right[$ such that 
$\mu \left( \factor {\ceP \ceQ}{0}{i} \right) =
\mu \left( \factor {\ceP \ceQ}{0}{j} \right).$
This contradicts the simplicity of $\ceP \ceQ$.
\end{proof}

\section{Proof of Theorem~\ref{RefLmMain1}}
\label{RefSectLmMain1}
For the rest of the paper, we will mostly restrict ourselves to the 
case of strings in $ \letters_2^* = \left\{ a, \conj a, b, \conj b \right\}^*$.
Let us first introduce the set of all the possible pairs satisfying the hypotheses of Theorem~\ref{RefLmMain1}:

\begin{center}
$
\candidatempty \left( \x  \right) := \{ 
\left( p_0, p_1 \right) \in \climple \sdiff \trapezoidals{} \sdiff \selffact{}. \ 
\forall x, i, j. \ 
\left[ i, j \right] \in 
\argmin_{\card} \bumps^x \left( p_0, p_1 \right) \ra{} 
\left( p_0, p_1 \right) - \left\{ i, j \right\} \notin \climple \sdiff \trapezoidals
\}.
$
\end{center}

Informally, Theorem~\ref{RefLmMain1} tells us that if the cancelling of all the minimal bumps breaks either simplicity of  $p_0 p_1$or the property $ \left( p_0, p_1 \right) \notin \trapezoidals$, then one of the bi-dimensional walks $\iso {p_0}$ and $\iso {p_1}$ has a detour of radius $1$.
While it is not obvious that such a detour exists, it is intuitive that such a detour can always be shortcut without compromising simplicity.
However, Theorem~\ref{RefLmMain1} additionally tells us that also the property $ \left( p_0, p_1 \right) \notin \trapezoidals$ is preserved after shortcutting.
To give an overview of the proof, we illustrate a classification of the possible reasons why cancelling a minimal bump from $ \left( p_0, p_1 \right) \in \candidatempty{}$ breaks the property of being in $\climple \sdiff \trapezoidals$.
Let us take, to fix our ideas, a minimal bump of indices $\left[ i, j \right]$, belonging to $p_0$, and having the form $b a^{m+2} \conj b$ for some $m \in \N$ (note that cancelling a bump of length less than $4$ cannot break simplicity).
Consider the points of $\Z^2$ which are ``underneath'' the bump (marked in Figure~\ref{RefFigBumps}, I)): 
it is clear that a necessary condition for 
$ \left( p_0, p_1  \right) - \left\{ i, j \right\} $ 
to be in $\climple \cap \trapezoidals$ is that the walk $\iso {p_0 p_1}$ does not touch any of those points, but the canceling of bump $\left[ i, j \right]$ changes 
the first or the last letter of $p_0$ or $p_1$, which means that
the first or the last point touched by the walk given by the bump $\left[ i, j \right]$ is in $ \left\{ 0, \mu \left( p_0 \right) \right\}$.
This configuration is depicted in for the case of $\mu \left( p_0 \right)$ (whose location is represented by a square) in Figure~\ref{RefFigBumps}~II).
Another possibility implying that the cancelling of $\left[ i, j \right]$ from $ \left( p_0, p_1 \right)$ breaks the property of being in $\climple \sdiff \trapezoidals$ is that 
$ \left( p_0, p_1  \right) - \left\{ i, j \right\} $ 
is no longer in $\climple$: this happens when the walk $\iso {p_0 p_1}$ touches at least one of the points marked in Figure~\ref{RefFigBumps}~I.
One configuration which is prevented when $\left[ i, j \right]$ is minimal is the one depicted in 
Figure~\ref{RefFigBumps}~III, where another bump 
$ \left[ i', j' \right] \in
\bumps^b \left( p_0, p_1 \right)$ 
manages to occupy some of those points having strictly smaller length than $\left[ i, j \right]$:
this should make clear that the points marked in Figure~\ref{RefFigBumps}~I cannot belong to another bump in $\bumps^b \left( p_0, p_1 \right)$.
This rules out a number of possibilities, and should help the reader to convince herself that only three configurations remain possibly making 
$ \left( p_0, p_1 \right) - \left\{ i, j \right\} \notin \climple$.
The first is when there is a bump $\left[ i', j' \right] \in \bumps^b \left( p_0 p_1 \right) \sdiff \bumps^b \left( p_0, p_1 \right)$ shorter than $\left[ i, j \right]$ (Figure~\ref{RefFigBumps}~IV).
\begin{figure}[]
\includegraphics[trim={3mm 5mm 67mm 78mm}, clip, scale=1.5]{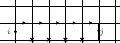}
I)
\hfill
\includegraphics[trim={3mm 22mm 67mm 60mm}, clip, scale=1.5]{f01}
II), type \clearbumps 
\hfill
\includegraphics[trim={3mm 13mm 67mm 70mm}, clip, scale=1.5]{f01}
III)
\\
\includegraphics[trim={3mm 44mm 67mm 38mm}, clip, scale=1.5]{f01}
IV), type \cloggedbumps 
\hfill
\includegraphics[trim={3mm 32mm 67mm 53mm}, clip, scale=1.5]{f01}
V), type \cornerbumps 
\hfill
\includegraphics[trim={3mm 37mm 67mm 45mm}, clip, scale=1.5]{f01}
VI), type \doublebumps 
\caption{Configurations of a minimal bump $\left[ i, j \right] \in \bumps^b \left( p_0, p_1 \right)$. $p_0$ goes from the circle (the origin) to the square ($\mu \left( p_0 \right)$), while $p_1$ goes from the square to the circle. }
\label{RefFigBumps}
\end{figure}

The other two are visualised in Figure~\ref{RefFigBumps}~V~and~VI.
Each of the configurations in Figure~\ref{RefFigBumps} II, IV, V and VI admits a number of symmetric variations (see the discussion on the transformational invariants at the beginning of Section~\ref{RefSectProofSkeleton}).
We will denote with $\clearbumps_1 \left( p_0, p_1 \right)$ the family of configurations represented by that in Figure~\ref{RefFigBumps}~II, with 
$\cloggedbumps_1 \left( p_0, p_1 \right)$ the family of configurations represented by that in Figure~\ref{RefFigBumps}~IV, with 
$\cornerbumps_1 \left( p_0, p_1 \right)$ the family of configurations represented by that in Figure~\ref{RefFigBumps}~V, 
and with $\doublebumps_1 \left( p_0, p_1 \right)$ the ones represented by that in Figure~\ref{RefFigBumps}~VI.
Further variations can be obtained by substituting in Figures~\ref{RefFigBumps} II, IV, V and VI the square representing $\mu \left( p_0 \right)$ with the circle representing the origin: this formally corresponds to swapping $p_0$ with $p_1$.
The configurations we obtain in this way from
$\clearbumps_1 \left( p_0, p_1 \right)$,
$\cloggedbumps_1 \left( p_0, p_1 \right)$, 
$\cornerbumps_1 \left( p_0, p_1 \right)$, 
and
$\doublebumps_1 \left( p_0, p_1 \right)$ 
will be called, respectively,
$\clearbumps_0 \left( p_0, p_1 \right)$,
$\cloggedbumps_0 \left( p_0, p_1 \right)$, 
$\cornerbumps_0 \left( p_0, p_1 \right)$,
 and
$\doublebumps_0 \left( p_0, p_1 \right)$. 
In the next subsection, we will give formal definitions for
$\clearbumps$,
$\cloggedbumps$,
$ \cornerbumps$,
$ \doublebumps$,
and will show that any minimal bump for $\left( p_0, p_1 \right)$ must indeed fall into one of these families, as soon as $ \left( p_0, p_1 \right) \in \candidatempty$.
First, however, we state
the following result, formalising the intuition that, if $\iso p$ 
``first goes up and then goes down'', then $\bumps^b \left( p \right) \neq \ze$; 
recall that $\proj_j$ is the projector extracting the $j$-th component of a tuple, 
while $\enum \left( m \right)$ is the $m$-th letter of $\letters^{\left( +  \right) }$.
\begin{Prop}
\label{RefLmUpDown}
Be given $p \in \letters_N^*$, $m \in \left[ 1, \ldots, N \right]$, and $i$, $j$, $k$ such that $j \in \left] i, k \right[ \subseteq \left[ 0, \left| p \right| \right]$.
\begin{align*}
\text{If } 
\left( \proj_m \circ \iso p  \right) \left( i \right) < \left( \proj_m \circ \iso p \right) \left( j \right) 
\text{ and }
\left( \proj_m \circ \iso p \right) \left( j \right) > \left( \proj_m \circ \iso p  \right) \left( k \right)
\text{, then } 
\left[ i, k-1 \right] \cap
\bumps^{\enum \left( m \right)} \left( p \right) \neq \ze;
\\
\text{if } 
\left( \proj_m \circ \iso p  \right) \left( i \right) > \left( \proj_m \circ \iso p \right) \left( j \right) 
\text{ and }
\left( \proj_m \circ \iso p \right) \left( j \right) < \left( \proj_m \circ \iso p  \right) \left( k \right)
\text{, then } 
\left[ i, k-1 \right] \cap
\bumps^{\conj{\enum \left( m \right)}} \left( p \right) \neq \ze.
\end{align*}
\end{Prop}

\begin{proof}
Note that the conditions impose that $\factor {p} {i} {k-1}$ is not short, so that it must have a detour; then use the definition of bump. 
\end{proof}

\subsection{Classification of minimal bumps}

The goal of this subsection is to show that the classification of minimal bumps depicted in Figure~\ref{RefFigBumps}~II, IV, V, and VI is exhaustive, in the sense that any minimal bump of a pair in $ \candidatempty$ falls into one of those categories (Lemma~\ref{RefLmBumpClassification4}).
We start from a definition capturing the configuration depicted in Figure~\ref{RefFigBumps}~III.

\begin{Def}
Given $ \left[ i_0, j_0 \right], \left[ i_1, j_1 \right] \in \bumps^{\alpha} \left( p \right)$, we say that $\left[ i_1, j_1 \right]$ \emph{nests $ \left[ i_0, j_0 \right]$ in $p$} iff
$ \iso p \left( \left[ i_0 + 1, j_0 \right] \right) + 
\left\{  \mu \left( \alpha  \right) \right\}
\subseteq \iso p \left(  \left[ i_1+2, j_1 - 1 \right]
 \right)$.
Often, it will be safe to just say that $\left[ i_1, j_1  \right]$ nests $ \left[ i_0, j_0 \right]$, without explicitly specifying $p$.
\end{Def}

Now we introduce definitions for the families of bumps 
$ \clearbumps_m$, $\cloggedbumps_m$, $\cornerbumps_m$ and $\doublebumps_m$ 
(with $m \in \left\{ 0, 1 \right\}$) already informally described at the end of last subsection.
We start with helper definitions.

\begin{Def}
$
\cornerbumps^{\alpha} \left( p \right) :=
\\
\left\{ \left[ i, j \right] \in \bumps^{\alpha} \left( p \right).
\left( i = 0 \wedge p \left( \left| p \right| - 1 \right) = \conj{p \left( 1 \right)} \right)
\vee
\left( j = \left| p \right| - 1 \wedge p\left( 0 \right) = \conj{p \left( \left| p \right| - 2 \right)} \right)
\right\}.
$
\\
$
\cloggedbumps^{\alpha} \left( p \right) := \{ \left[ i, j \right] \in \bumps^{\alpha} \left( p \right). 
\exists m, n \in \N, \beta \ortho \alpha, p_1.
\\
\left( 
p = \beta^n \conj \alpha p_1 \alpha \beta^m \wedge m+n > 0 \wedge
\left( 
\iso p \left[ 0, n \right] \cup \iso p \left[ \left| p \right| - m, \left| p \right| \right]  \right) 
+ \left\{ \iso \alpha \right\} 
\subseteq 
\iso p \left[ i+2, j-1  \right]
\right)
\}$
\end{Def}

The helper definitions above allow us to introduce 
$ \clearbumps_m$, $\cloggedbumps_m$, $\cornerbumps_m$ and $\doublebumps_m$ 
($m \in \left\{  0, 1\right\}$) as follows.
\begin{Def}
\label{RefDefCornerBumps2}
$
\clearbumps_0^{\alpha} \left( p_0, p_1 \right) :=
\left\{ \left[ i, j \right] \in \bumps^{\alpha} \left( p_0, p_1 \right). \ 
\left( p_0, p_1 \right) - \left\{ i, j \right\} \in \climple \wedge
\left( i=0 \vee j=\left| p_0 p_1 \right| - 1 \right)
\right\}
$
\\
$
\clearbumps_1^{\alpha} \left( p_0, p_1 \right) :=
\left\{ \left[ i, j \right] \in \bumps^{\alpha} \left( p_0, p_1 \right). \ 
\left( p_0, p_1 \right) - \left\{ i, j \right\} \in \climple \wedge
\left( i=\left| p_0 \right| \vee j=\left| p_0 \right| - 1 \right)
\right\}
$
\\
$ 
\cornerbumps_0^{\alpha} \left( p_0, p_1 \right) :=
\cornerbumps^{\alpha} \left( p_0 p_1 \right) \cap \bumps^{\alpha} \left( p_0, p_1 \right).
$
\\
$ 
\cornerbumps_1^{\alpha} \left( p_0 , p_1 \right) :=
\left\{ \left[ i,j \right] \in \bumps^{\alpha} \left( p_0, p_1 \right).
\left\{ \left[ i-1, i+1 \right], \left[ j-1, j+1 \right] \right\} \cap 
\left( \bumps^{\alpha^\ortho} \left( p_0 p_1 \right) \sdiff \bumps^{\alpha^\ortho} \left( p_0, p_1 \right) \right) \neq \ze{}
\right\}
$
\\
$
\cloggedbumps_0^{\alpha} \left( p_0, p_1 \right):=
\cloggedbumps^{\alpha} \left( p_0 p_1 \right) \cap \bumps^{\alpha} \left( p_0, p_1 \right)
$
\\
$
\cloggedbumps_1^{\alpha} \left( p_0, p_1 \right) := 
\left\{ \left[ i, j \right] \in \bumps^{\alpha} \left( p_0, p_1 \right).
\text{ there is } \left[ i_0, j_0 \right] \in \bumps^{\alpha} \left( p_0 p_1 \right) \sdiff \bumps^{\alpha} \left( p_0, p_1 \right) \text{ nested by } \left[ i, j \right]
\right\}
$
\\
$
\doublebumps_0^{\alpha} \left( p_0, p_1 \right) := 
\{ \left[ i, j \right] \in \bumps^{\alpha} \left( p_0, p_1 \right).
\\
\left( i = 1 \wedge \left[ 0, 2 \right] \in \bumps_0^{\alpha^\ortho} \left( p_0, p_1 \right) \right) \vee \left( j = \left| p_0 p_1 \right| - 2 \wedge \left[ \left| p_0 p_1 \right| -3, 
\left| p_0 p_1 \right| - 1
\right] \in \bumps_1^{\alpha^\ortho} \left( p_0, p_1 \right)\right)
\}
$
\\
$
\doublebumps_1^{\alpha} \left( p_0, p_1 \right) :=
\{ \left[ i, j \right] \in \bumps^{\alpha} \left( p_0, p_1 \right).
\\
\left( i = \left| p_0 \right| + 1 \wedge \left[ \left| p_0 \right|, \left| p_0 \right| + 2 \right] \in \bumps_1^{\alpha^\ortho} \left( p_0, p_1 \right) \right)
\vee 
\left( j = \left| p_0 \right| - 2 \wedge \left[ \left| p_0 \right| -3, 
\left| p_0 \right| - 1
\right] \in \bumps_0^{\alpha^\ortho} \left( p_0, p_1 \right)\right)
\}
$
\end{Def}

\begin{Not}
We will write $\bumps^X \left( p \right)$ for $\bigcup_{x \in X} \bumps^x \left( p \right)$, 
$\bumps \left( p \right)$ for $\bumps^{\letters} \left( p \right)$, 
and similarly for $\bumps^X \left( p_0, p_1 \right)$, $\clearbumps^X_0 \left( p_0, p_1 \right)$,
$ \bumps \left( p_0, p_1 \right)$, $ \clearbumps_0 \left( p_0, p_1 \right)$
and so on.
Also, we will write $ \clearbumps^X \left( p_0, p_1 \right)$ for
$\clearbumps^X_0 \left( p_0, p_1 \right) \cup \clearbumps^X_1 \left( p_0, p_1 \right)$, and so on.
Finally, we will often use the following more suggestive notations:
$\bumps^{\Right}$ for $\bumps^a$, $\bumps^{\Left}$ for $\bumps^{\conj a}$,
$\bumps^{\Up}$ for $\bumps^b$ and $\bumps^{\Down}$ for $\bumps^{\conj b}$.
Similarly for $\clearbumps$ and the other notions introduced in Definition~\ref{RefDefCornerBumps2}.
When it is clear from the context, the argument can be dropped, writing, e.g., $\bumps^{\alpha}$ in lieu of $\bumps^{\alpha} \left( p_0, p_1 \right)$. The next result is by symmetry.
\end{Not}

\begin{Prop}
\label{RefLmSymmetry}
{\small{}
For $ \left( p_0, p_1 \right) \in \climple$, 
$
\cloggedbumps_0^{\alpha} \left( p_0, p_1 \right) = \ze 
\to {}
\cloggedbumps_1^{\alpha} \left( p_1, p_0 \right) = \ze
$
and
$
\cornerbumps_0^{\alpha} \left( p_0, p_1 \right) = \ze 
\to {}
\cornerbumps_1^{\alpha} \left( p_1, p_0 \right) = \ze
$.
}
\end{Prop}

\begin{Prop}
\label{RefLmCloggedShortcut}
Let $ \left( p_0, p_1 \right) \in \climple{}$, set $p := p_0 p_1$ and assume 
$ \cloggedbumps_0^{\alpha} \left( p_0 ,p_1 \right) \cap \bumps_k^{\alpha} \left( p_0 , p_1 \right) \neq \ze{}$ for some $\alpha \in \Sigma_2$, $k \in \left\{ 0, 1 \right\}$.
Then $ l := \iso p ^{-1} \left( \iso \alpha \right) $ 
is uniquely defined and belongs to 
$ k \left| p_0 \right| + \left[ 3 - k, \left| p_k \right| - 2 - k \right]$.
Moreover, if we define
\begin{align*}
\left( i_0, j_0, q_0 \right) := 
\left\{  
\begin{aligned}
&
\left( 0, l - 1, \alpha \right) 
& \text{ if } k=0
\\
&
\left( l - \left| p_0 \right|, \left| p_1 \right| - 1, \conj \alpha \right)
& \text{ if } k=1,
\end{aligned}
\right.
\end{align*}
then $ \left( i_0, j_0, q_0 \right)$ is a $1$-shortcut of $p_k$.
\end{Prop}

\begin{proof}
Consider $ \left[ i, j \right] \in \cloggedbumps_0^{\alpha} \cap \bumps_k^{\alpha}$.
By definition of $\cloggedbumps_0^{\alpha}$, $\iso \alpha \in \iso p \left[ i + 2, j - 1 \right]$.
By simplicity of $p$, this implies that $l$ is uniquely defined and lies in 
$ \left[ i+2, j-1 \right]$.
Moreover, the definition of $\cloggedbumps_0$ also implies $i > 0$ and $j < \left| p \right| - 1$, so that
\begin{align}
\label{RefFormula13}
4 \leq i+3 \leq j 
&&
l \in \left[ 3, \left| p \right| -3 \right].
\end{align}

When $k=0$, $\iso p_0 \left( l \right) = \iso \alpha$ and $\left| p_0 \right| \geq j$, so that, using~\eqref{RefFormula13},  
$ \factor {p_0}{0}{l-1} \sim \alpha$ 
and 
$ \left| \factor{p_0}{0}{l-1} \right| \geq 
 \left| \factor{p_0}{0}{2} \right| = 3
> \left| \alpha \right|$.
When $k=1$, $i \geq \left| p_0 \right|$, $\iso p \left( l \right) = \iso p_0 + \iso p_1 \left( l - \left| p_0 \right| \right) = \iso \alpha$ and $\iso p_0 + \iso p_1 \left( \left| p_1 \right| \right) = \ve 0$, so that
$ \factor{p_1}{l - \left| p_0 \right|}{\left| p_1 \right| - 1} \sim \conj \alpha$.
Furthermore,  using~\eqref{RefFormula13} yields $l \leq \left| p \right| - 3 = \left| p_0 \right| + \left| p_1 \right| - 3$, so that
$ 0 \leq l - \left| p_0 \right| \leq \left| p_1  \right| - 3$ and 
$\left| \factor{p_1}{l - \left| p_0 \right|}{\left| p_1 \right| - 1} \right| 
\geq \left| \factor{p_1}{\left| p_1 \right| - 3}{\left| p_1 \right| - 1} \right| 
= 3 > \left| \conj \alpha \right|$.
\end{proof}

\begin{Prop}
\label{RefLmNest2}
Let $p \in \Sigma_2^*$ be simple, and $ \left[ i_1, j_1 \right], \left[ i_0, j_0 \right] \in \bumps^{\alpha} \left( p \right)$. 
Suppose that $\iso p \left[ i_0+1, j_0 \right] + \alpha \cap \iso p \left[ i_1 + 2, j_1 - 1 \right] \neq \ze$.
Then $ \left[ i_1, j_1 \right]$ nests $ \left[ i_0, j_0 \right]$.
\end{Prop}

\begin{proof}
It must be $\left| p \right| > 2$, $i_0 + 1 < j_0$, $i_1+2 < j_1 $. 
Consider the functions 
$ f_0 : \left[ i_0 + 1, j_0 \right] \ni n \mapsto \iso p \left( n \right)$ and 
$ f_1 : \left[ i_1 + 2, j_1 - 1 \right] \ni n \mapsto \iso p \left( n \right) - \iso \alpha$.
Without loss of generality, we set $\alpha := b$, so that 
$ \ran \left( \pi_2 \circ f_0 \right) = \ran \left( \pi_2 \circ f_1 \right) = \left\{ y \right\}$ for some $y \in \Z$.
$\ran \left( \pi_1 \circ f_0 \right)$ and
$\ran \left( \pi_1 \circ f_1 \right)$ being integer intervals, we can set
$ \left[ k_0, l_0 \right] := \ran \left( \pi_1 \circ f_0 \right)$ 
and
$ \left[ k_1, l_1 \right] := \ran \left( \pi_1 \circ f_1 \right)$.
If it were $k_0 < k_1$, then it would be $l_0 \geq l_1$ due to the hypothesis.
Hence, $k_1 - 1$ would be in $ \left[ k_0, l_0 \right]$, so that $ \left( k_1 - 1, y \right) = 
\iso p \left( j_1 \right) \in \ran f_0 = \iso p \left[ i_0 + 1, j_0 \right]$.
This would imply, by simplicity, that $j_0 = \left| p \right|$ and $j_1 = 0$, which is impossible.
Analogously, one shows $l_0 \leq l_1$.
Ultimately, $\ran f_0 \subseteq \ran f_1$, which is the thesis.
\end{proof}

\begin{Prop}
\label{RefLmBumpClassification1}
Let $p \in \climple{}$, and 
$ \left[ i_1, j_1 \right] \in \bumps^{\alpha} \left( p \right) \sdiff \cornerbumps^{\alpha} \left( p \right)$.
Assume that
\begin{enumerate}
\item
\label{RefAssNoOrthBumps}
$ \left\{ \left[ i_1 - 1, i_1+1 \right], \left[ j_1 - 1, j_1 + 1 \right] \right\} \cap \bumps^{\alpha^{\ortho}} \left( p \right) = \ze{}$
and that
\item
\label{RefAssNotSimple}
$p - \left\{ i_1, j_1 \right\}$ is not simple.
\end{enumerate}
Then either 
$ \left[ i_1, j_1 \right] \in \cloggedbumps^{\alpha} \left( p \right)$ 
or 
$ \left[ i_1, j_1 \right]$ nests some 
$\left[ i_0, j_0 \right] \in \bumps^{\alpha} \left( p \right)$. 
\end{Prop}

\begin{proof}
From hypothesis~\ref{RefAssNotSimple}, 
$ \left| p \right|>2$, $j_1 > i_1 + 2$ and 
$ \factor{p}{i_1}{j_1} = \alpha \beta^{j_1 - i_1 - 1} \conj \alpha $ for some 
$\beta \ortho \alpha$.
Set 
\\
$M:= \iso p ^{-1} \left( \iso p \left[ i_1 + 2, j_1 - 1 \right] - \iso \alpha \right).$
\\
$M$ is not empty because $ p - \left\{ i_1, j_1 \right\}$ is not simple.
Define
$
h : M \ni i \mapsto \bigcup \left\{ J \subseteq M. J \text{ is an interval } \wedge i \in J \right\}.
$
$ \ran h$ is the coarsest partition of $M$ among those consisting of intervals.
As such, it is non empty.
We consider a generic $\left[ i, j \right] \in \ran h$, and prove the following facts:
\begin{align}
\label{RefFormula07}
i \leq j < i_1 - 1 < j_1 \text{ or } 
i_1 < j_1 < j_1 + 2 < i \leq j
\\
\label{RefFormula08}
i - 1 \in \dom {\iso p} \ra{} p \left( i - 1 \right) = \alpha
\\
\label{RefFormula09}
j + 1 \in \dom {\iso p} \ra{} p \left( j \right) = \conj \alpha.
\end{align}
\begin{description}
\item[Proof of \eqref{RefFormula07}: ]
since 
$ \iso p \left[ i_1+2, j_1 - 1 \right] - \iso \alpha \cap 
\iso p \left[ i_1 , j_1 + 1 \right] = \ze,$ 
it must be 
$ \left[ i_1, j_1 + 1 \right] \cap \left[ i, j \right] = \ze$, 
which means that either 
$i \leq j < i_1 < j_1$ or 
$ i_1 < j_1 + 1 < i \leq j$.
This can be strengthened to the wanted inequality: indeed, if it were $j = i_1 - 1$, then 
$\iso p \left( i_1 \right) - \iso p \left( i_1 - 1 \right) = - \iso \beta$, so that 
$ \factor p {i_1 - 1} {i_1 + 1} = \conj \beta \alpha \beta$, which is prevented by hypothesis~\ref{RefAssNoOrthBumps}.
Similarly, $i$ can be shown being not equal to $j_1 + 2$.
\item[Proof of \eqref{RefFormula08}: ]
we note that it cannot be $p \left( i - 1 \right) = \conj \alpha$, for then it would be
$ \iso p \left( i - 1 \right) \in 
\iso p \left[ i_1 + 2, j_1 - 1 \right]$, which would imply $ i - 1 \in \left[ i_1 + 2, j_1 - 1 \right]$ by simplicity, thus violating \eqref{RefFormula07}.
By exclusion, it hence suffices to show that it cannot be $p \left( i - 1 \right) // \beta$, which we do by contradiction.
$ \iso p \left( i - 1 \right) \in \iso p \left[ i_1 + 1, j_1 \right] - \iso \alpha$, while, by construction of $h$, it cannot be 
$ \iso p \left( i - 1 \right) \in \iso p \left[ i_1 + 2, j_1 - 1 \right] - \iso \alpha$, so that
$ \iso p \left( i - 1 \right) \in \iso p \left\{ i_1 + 1, j_1 \right\} - \iso \alpha = 
\iso p \left\{ i_1, j_1 + 1 \right\}$.
Now, it cannot be $i - 1 \in \left\{ i_1, j_1 + 1 \right\}$ due to~\eqref{RefFormula07}, hence the only possibility is that $ \left\{ 0, \left| p \right| \right\}$ includes $i - 1$, together with one element of $\left\{ i_1, j_1 + 1 \right\} $.
Since $i \in \dom \iso p$, this implies $i - 1 = 0$, and hence $j_1 + 1 = \left| p \right|$.
Ultimately, we draw $p \left( \left| p \right| - 1 \right) = \conj \alpha$, $p \left( \left| p \right| - 2 \right) = \beta$, and $ p \left( 0 \right) = \conj \beta$.
The last fact follows from noting that
$ \iso p \left( i \right) \in \iso p \left[ i_1 + 2, \left| p \right| - 2 \right] - \iso \alpha 
\subseteq
\left\{ - n \iso \beta. n \in \Z^+ \right\} 
.
$
This violates the hypothesis $ \left[ i_1, j_1 \right] \notin \cornerbumps^{\alpha} \left( p \right)$.
\item[Proof of \eqref{RefFormula09}: ]
Analogous to that of~\eqref{RefFormula08}. 
Now we proceed by cases.
\end{description}
\begin{description}
\item[There is {$ \left[ i_0, j_0 \right] \in \ran h$} such that $i_0 - 1 \in \dom \iso p$ and 
$ j_0 + 1 \in \dom \iso p:$ ]
then, using~\eqref{RefFormula08} and \eqref{RefFormula09},
\begin{align*}
p \left( i_0 - 1 \right) = \alpha && \text{ and } && 
p \left( j_0  \right) = \conj \alpha,
\end{align*}
which implies $i_0 < j_0$ by simplicity.
Moreover, for every $i = i_0 \ldots j_0$, $\iso p \left( i \right) \in M$, so that 
for every $i = i_0 \ldots j_0 - 1$, $p \left( i \right)$ must be parallel to $\beta$.
This means that $ \left[ i_0, j_0 \right] \in \bumps^{\alpha} \left( p \right)$ and hence that $ \left[ i_1, j_1 \right]$ nests $\left[ i_0, j_0 \right]$ by definition of $M$.
\item[Remaining cases: ] consider 
$ \left[ i, j \right] \in \ran h \neq \ze$. 
Since we are excluding the previous case, it must be either 
$i = 0$ or $j = \left| p \right|$, 
which, $p$ being closed, implies that 
$ \left\{ 0, \left| p \right| \right\} \subseteq M$;
given that $ M \subset \left[ 0 , \left| p \right| \right]$, this yields that
$ \left\{ \left[ 0, j_0 \right], \left[ \left| p \right| - i_0  , \left| p \right| \right] \right\} \subseteq \ran h$ 
for some $i_0, j_0 \in \N$, with $ j_0 < \left| p \right| - i_0 $.
Since $ \left[ 0, j_0 \right] \subseteq M$, it must be $ p \left( i \right) \ortho \alpha$ for every $i = 0 \ldots j_0 - 1$, while $ p\left( j_0 \right) = \conj \alpha$ by~\eqref{RefFormula09};
$p$ being simple, $\beta_1^{j_0} \alpha$ must be a left factor of $p$ for some $\beta_1 \ortho \alpha$.
Analogously, $ p  \left[ \left| p \right| - i_0, \left| p \right| - 1 \right] \ortho \alpha$, and $ p\left( \left| p \right| - i_0 - 1 \right) = \alpha$ using~\eqref{RefFormula08}, so that 
$\conj \alpha \beta_2^{i_0}$ must be a right factor of $p$ for some $\beta_2 \ortho \alpha$.
Now, since $\left| p \right| > 2$ and $p$ is closed, it cannot be $ j_0 + i_0 = 0$, otherwise simplicity of $p$ would fail; this in turn implies that we can take $\beta_2 = \beta_1$.
$ \left[ i_1, j_1 \right] \in \cloggedbumps^{\alpha} \left( p \right)$ by definition.
\end{description}
\end{proof}

\begin{Prop}
\label{RefLmBumpClassification2}
Let $\left( p_0, p_1 \right) \in \climple$, $ \left[ i_1, j_1 \right] \in \bumps^{\alpha}\left( p_0, p_1 \right)$ and assume that 
\begin{enumerate}
\item
\label{RefAssNoBumps2}
$\left\{ \left[ i_1 - 1, i_1 + 1 \right], \left[ j_1 - 1, j_1 + 1  \right] \right\} \cap
\bumps^{\alpha^{\ortho}} \left( p_0, p_1 \right) = \ze$
\item
$ 
\left( p_0, p_1 \right) - \left\{ i_1, j_1 \right\} \notin \climple{}.
$
\end{enumerate}
Then either 
$ \left[ i_1, j_1 \right] \in \cornerbumps^{\alpha} \left( p_0, p_1 \right)
\cup \cloggedbumps^{\alpha} \left( p_0, p_1 \right)
$ or there is $ \left[ i_0, j_0 \right] \in \bumps^{\alpha} \left( p_0, p_1 \right)$ nested by $ \left[ i_1, j_1 \right]$.
\end{Prop}

\begin{proof}
Let $p := p_0 p_1$ and assume $ \left[ i_1, j_1 \right] \notin \cornerbumps_1^{\alpha} \left( p_0, p_1 \right)$.
Then, by definition~\ref{RefDefCornerBumps2} and hypothesis~\ref{RefAssNoBumps2},
\\
$\left\{ \left[ i_1 - 1, i_1 + 1 \right], \left[ j_1 - 1, j_1 + 1  \right] \right\} \cap
\bumps^{\alpha^{\ortho}} \left( p \right) = \ze$.
After applying Proposition~\ref{RefLmBumpClassification1}, there are three cases:
\begin{description}
\item[Case {$ \left[ i_1, j_1 \right] \in \cornerbumps^{\alpha} \left( p \right) \cup \cloggedbumps^{\alpha} \left( p \right): $}]
then $ \left[ i_1, j_1 \right] \in \cornerbumps_0^{\alpha} \left( p_0, p_1 \right) \cup 
\cloggedbumps_0^{\alpha} \left( p_0, p_1 \right)
$ by definition.
\item[ There is {$ \left[ i_0, j_0 \right] \in \bumps^{\alpha} \left( p_0, p_1 \right)$ nested by $ \left[ i_1, j_1 \right]$}: ]
immediate.
\item[ There is {$ \left[ i_0, j_0 \right] \in \bumps^{\alpha} \left( p \right) \sdiff \bumps^{\alpha} \left( p_0, p_1 \right)$ nested by $ \left[ i_1, j_1 \right]$}: ]
then $ \left[ i_1, j_1 \right] \in \cloggedbumps^{\alpha}_1 \left( p_0, p_1 \right)$ by definition.
\end{description}
\end{proof}

\begin{Prop}
\label{RefLmBumpClassification3}
Let $\left( p_0, p_1 \right) \in \climple \sdiff \trapezoidals{}$, $ \left[ i_1, j_1 \right] \in \bumps^{\alpha}\left( p_0, p_1 \right)$. 
Assume that \\
\begin{enumerate*}
\item
$\left\{ \left[ i_1 - 1, i_1 + 1 \right], \left[ j_1 - 1, j_1 + 1  \right] \right\} \cap
\bumps^{\alpha^{\ortho}} \left( p_0, p_1 \right) = \ze$ and 
\item
$ 
\left( p_0, p_1 \right) - \left\{ i_1, j_1 \right\} \notin \climple{} \sdiff \trapezoidals{}.
$
\end{enumerate*}
Then either 
$ \left[ i_1, j_1 \right] \in \clearbumps^{\alpha} 
\left( p_0, p_1 \right)
\cup \cornerbumps^{\alpha} \left( p_0, p_1 \right)
\cup \cloggedbumps^{\alpha} \left( p_0, p_1 \right)
$ or there is $ \left[ i_0, j_0 \right] \in \bumps^{\alpha} \left( p_0, p_1 \right)$ nested by $ \left[ i_1, j_1 \right]$.
\end{Prop}

\begin{proof}
Thanks to Proposition~\ref{RefLmBumpClassification2} and hypotheses, we can assume
$ \left( p_0, p_1 \right) - \left\{ i_1, j_1 \right\} \in \climple \cap \trapezoidals{}$.
In the case $j_1 < \left| p_0 \right|$, this implies that the first or the last letter of $ p_0 - \left\{ i_1, j_1 \right\}$ (or both) are different than those of $p_0$; 
this can only happen when $ \left\{ 0, \left| p_0 \right| - 1 \right\} \cap \left\{ i_1, j_1 \right\} \neq \ze{}$, yielding 
$ \left[ i_1, j_1 \right] \in \clearbumps^{\alpha} \left( p_0, p_1 \right)$ by definition.
The other case, $ i_1 \geq \left| p_0 \right|$, is similar.
\end{proof}

\begin{Lm}
\label{RefLmBumpClassification4} 
Let $\left( p_0 , p_1 \right) \in \candidatempty$, and 
$ \left[ i, j \right] \in \argmin_{\size} \bumps^{\alpha} \left( p_0, p_1 \right)$ for some $\alpha$.
Then $ \left[ i, j \right] \in 
\clearbumps^{\alpha} \left( p_0, p_1 \right) 
\cup \doublebumps^{\alpha}  \left( p_0, p_1 \right)
\cup \cornerbumps^{\alpha}  \left( p_0, p_1 \right)
\cup \cloggedbumps^{\alpha} \left( p_0, p_1 \right).$
\end{Lm}

\begin{proof}
By cases.
\begin{description}
\item[Case 
{$ \left\{ \left[ i-1, i+1 \right], \left[ j-1, j+1 \right] \right\} \cap \bumps^{\alpha^\ortho} \left( p_0, p_1 \right) 
\neq \ze
$}: ]
consider $k \in \left\{ i, j \right\}$ such that 
$ \left[ k-1, k+1 \right] \in \bumps^{\alpha^\ortho}\left( p_0, p_1 \right)$.
Note that $ p_0 p_1 - \left\{ k-1, k+1 \right\}$ is still simple; moreover, no element of 
$ \bumps \left( p_0, p_1 \right)$ has cardinality less than $3 = \card \left[ k-1, k+1 \right]$, because $\left( p_0, p_1 \right) \in \climple \sdiff \trapezoidals \sdiff \selffact$.
Hence, by definition of $\candidatempty{}$, 
$ \left( p_0, p_1 \right) - \left\{ k-1, k+1 \right\} \in \climple \cap \trapezoidals{}$.
Since $ \left( p_0, p_1 \right) \notin \trapezoidals{}$, this implies that\\
$ \left\{ k-1, k+1 \right\} \cap \left\{ 0, \left| p_0 \right| - 1, \left| p_0 \right|, \left| p_0 p_1 \right| - 1 \right\} \neq \ze$
$ \ra{} k-1 \in \left\{ 0, \left| p_0 \right| \right\} \vee 
k+1 \in \left\{ \left| p_0 \right| - 1, \left| p_0 p_1 \right|- 1 \right\}
$
\\
$ \ra{} \left( k = i \wedge i \in \left\{ 1, \left| p_0 \right| + 1 \right\} \right)
\vee
$
$
\left( k = j \wedge j \in \left\{ \left| p_0 \right| - 2, \left| p_0 p_1 \right| - 2 \right\} \right):
$
$ \left[ i, j \right] \in \doublebumps^{\alpha} $ by definition.
\item[Case 
{$ \left\{ \left[ i-1, i+1 \right], \left[ j-1, j+1 \right] \right\} \cap \bumps^{\alpha^\ortho} \left( p_0, p_1 \right) = \ze
$}: ]
by definition of $\candidatempty{}$, $ \left( p_0, p_1 \right) - \left\{ i,j \right\} \notin \climple \sdiff \trapezoidals{}$.
We can hence invoke Proposition~\ref{RefLmBumpClassification3} to obtain that 
$ \left[ i, j \right] \in \clearbumps^{\alpha} \cup \cornerbumps^{\alpha} \cup \cloggedbumps^{\alpha}$, because if there were a bump nested by $ \left[ i, j \right] $, it would have a strictly smaller cardinality.
\end{description}
\end{proof}

\subsection{Bumps in $ \clearbumps \cup \doublebumps \cup \cornerbumps{}$}

\begin{Prop}
\label{RefLmSimpleBumps01} 
Let $(p_0, p_1) \in \climple{}$, and assume
$
\forall \left[ i, j \right] \in \bumps^{\left\{ \beta, \conj \beta  \right\} } \left( p_0, p_1 \right). \ 
\card \left[ i, j \right] > 3.
$
Then, given $ \left[ i_1, j_1 \right] \in \bumps^{\alpha} \left( p_0, p_1 \right)$, $\beta \ortho \alpha$, either
$ \left[ i_1, j_1 \right]$ nests some $ \left[ i_0, j_0 \right] \in \bumps^{\alpha} \left( p_0, p_1 \right)$, 
or
$ p_0 p_1 - \left\{ i_1, j_1 \right\}$ is simple and closed, 
or
$ \left[ i_1, j_1 \right] \in 
\cornerbumps^{\alpha} \cup \cloggedbumps{}^{\alpha}.$
\end{Prop}
\begin{proof}
This is an immediate corollary of Proposition~\ref{RefLmBumpClassification2}: let us assume, by contradiction, that  
$ p_0 p_1 - \left\{ i_1, j_1 \right\}$ is not simple; this triggers the second hypothesis of Proposition~\ref{RefLmBumpClassification2}. 
The inequality hypothesis triggers the first one, and we obtain the thesis. 
\end{proof}

\begin{Prop}
\label{RefLmStraightPreserved}
Let $p \in \Sigma_2^*$ be simple, $\alpha \in \Sigma_2$, $i, j \in \Z^+$, and suppose $\bumps^{\alpha} \left( p \right) = \ze{}$. 
Then
\begin{align*}
\bumps^{\alpha} \left( p - i \right) = \ze, \text{ and }
\left( \left[ i, j \right] \in \bumps^{\beta}\left( p \right) \wedge \beta \ortho \alpha  \right) 
\ra{} \bumps^{\alpha}\left( p - \left\{ i, j \right\} \right) = \ze{}.
\end{align*}
\end{Prop}
\begin{proof}
Let us prove only the second thesis, since the first one is simpler and uses similar ideas. 
By simplicity, we can assume that $\factor p i j = b a^{m+1} \conj b$ for some $m \in \N$, so that $\beta = b$ and $\alpha \in \left\{ a, \conj a \right\}$. 
By contradiction, consider $\left[ h, k \right] \in \bumps^{\alpha}\left( q:=p - \left\{ i, j \right\} \right)$.
The interior $a^{m+1}$ of the canceled bump cannot be included between the positions $h$ and $k$, since this is a string whose alphabet is either $\left\{ b \right\}$ or $\left\{ \conj b \right\}$.
So it's either $q\left( k \right)$ coinciding with the first $a$ of that interior, or $h$ coinciding with its last $a$.
Let us take the first option (the second being analogous). 
Then the interior of $[h, k]$ must be made of $\conj b$, otherwise we would violate the hypothesis $\bumps^{\alpha} \left( p \right) = \ze{}$. 
But then $p$ would have a factor $\conj b b$, violating simplicity.
\end{proof}

\begin{Lm}
\label{RefLmTopological}
The set 
$\{ \left( p_0, p_1 \right) \in \climple{}. \ 
\exists i, j \in \Z^+, \alpha \ortho \beta . \  
\beta \alpha ^ i \conj \beta \text{ is a left factor of } p_0 \ \wedge $
\\
$\conj \beta \alpha^j \beta \text{ is a right factor of } p_0 \ 
\wedge \ 
\bumps^{\alpha, \conj \alpha} \left( p_0 \right) = \ze{} = \bumps^{\alpha, \conj \alpha} \left( p_1 \right) 
\}$
is empty. 
\end{Lm}

\begin{proof}
Suppose it is not, call it $M$ and consider an element $ \left( p_0, p_1 \right) $ of $M$ such that $ \left| p_0 p_1 \right|$ is minimal. 
We can assume $ p_0 = b \conj a^{i_0} \conj b p_2 = p_4 \conj b \conj a^{j_0} b$ for some $i_0, j_0 \in \Z^+$.
By definition of $M$, 
\begin{align}
\label{RefFormula05}
\bumps^{\Left, \Right} \left( p_0 \right) = \ze{} = 
\bumps^{\Left, \Right} \left( p_1 \right).
\end{align}
Let us set $p := p_0 p_1$, $\left( x_1, y_1 \right) := \iso p   \left( i_0 + 2 \right) = 
\left( - i_0, 0 \right)$, 
$ \left( x_2, y_2 \right) := \iso p \left( \left| p_4 \right| \right)
$,
and
$ \left( x_3, y_3 \right) := \iso p \left( \left| p_0 \right| \right)
$.
Note that it must be $0 > x_1 \geq x_2 > x_3$, because $ \bumps^{\Left} \left( p_0 \right) = \ze$. 
If it were $p_1 \left( 0 \right) = a$, then $j_0 > 1$ by simplicity, so that $
\left( 
p_0 - \left( \left| p_0 \right| - 2  \right)
, p_1 - 0
 \right)
$ would still be in $M$,
contradicting $p$'s minimality; on the other hand, it cannot be $ p_1 \left( 0 \right) = \conj a$, because $\bumps^{\Left} \left( p_1 \right) = \ze$: we must conclude $p_1 \left( 0 \right) = b$.
Indeed, we can strengthen this by noting that the first letter of $p_1$ different from $b$ cannot be $\conj a$, because $\bumps^{\Left} \left( p_1 \right) = \ze{}$, so that $p_1$ must admit $b^{i_5} a$ as a left factor, with $i_5 \in \Z^+$; 
this can be made even stronger by noting that the first letter of $p_1$ following the left factor $b^{i_5} a$ cannot be $\conj b$, due to the minimality of $p$ and Proposition~\ref{RefLmStraightPreserved}: $p_1$ admits $b^{i_5} a \gamma{}$, where $\gamma \in \Sigma_2 \sdiff \left\{ \conj b \right\}$.
Analogously, $p_1$ admits $ \delta{} a b^{j_5}$ as a right factor for some $j_5 \in \Z^+$, $\delta \in \Sigma_2 \sdiff \left\{ \conj b  \right\}$. 
By a similar minimality argument we also draw that $j_0 = i_0 = 1$; to recapitulate:
\begin{align}
\label{RefFormula06}
b^{i_5} a \gamma \text{ is a left factor of } p_1, && 
\delta a b^{j_5} \text{ is a right factor of } p_1, 
&&
i_0 = j_0 = 1.
\end{align}
Now, $y_2 = y_3$, so that we can apply~\ref{RefLmGeometric} to the pairs 
$ \left( \left( 0, 0 \right), \left( x_1, y_1 \right) \right)$ and 
$ \left( \left( x_2, y_2 \right), \left( x_3, y_3 \right) \right),$
obtaining that 
$\bumps \left( p_6 \right) \cup \bumps \left( p_1 \right) \neq \ze{}, $
where $p_6$ is the proper factor of $p_0$ joining the points $\left( x_1, y_1 \right)$ and $\left( x_2, y_2 \right)$.
Combined with \eqref{RefFormula05}, this tells us that the following set is non-empty:
\begin{align*}
\left\{ \left[ i, j \right] \in
\bumps^{\Up, \Down} \left( p_0, p_1 \right) 
. \left[ i, j \right] \subseteq \left[ i_1 , j_1 \right] \vee \left[ i, j \right] \subseteq \left[ \left| p_0  \right|, \left| p \right| - 1 \right]
 \right\}
,
\end{align*}
where $i_1, j_1$ are such that $ \factor{p_0}{i_1}{j_1} = p_6 $.
Hence, consider an element $\left[ i_2, j_2 \right]$ of it having minimal cardinality; 
to fix the ideas, we assume 
$ \left[ i_2, j_2 \right] \in \bumps^{\Up} \left( p_0, p_1 \right),$ 
with the proof for the case
$ \left[ i_2, j_2 \right] \in \bumps^{\Down} \left( p_0, p_1 \right)$ being similar.  
$p - \left\{ i_2, j_2 \right\}$ cannot be in $\climple$ for, if it were, 
$ \left( p_0, p_1 \right) - \left\{ i_2, j_2 \right\}$ would also be in $M$ by Proposition~\ref{RefLmStraightPreserved}.
Moreover, $ \cornerbumps \left( p_0, p_1 \right) = \ze$ due to~\eqref{RefFormula06}, and 
$ \cloggedbumps \left( p_0, p_1 \right) = \ze$ because of~\eqref{RefFormula05}.
By Proposition~\ref{RefLmSimpleBumps01}, then, $ \left[ i_2, j_2 \right]$ must nest some $ \left[ i_3, j_3 \right] \in \bumps^{\Up} \left( p_0, p_1 \right).$
Due to the minimality of $\left[ i_2, j_2 \right]$, 
$ \left[ i_3, j_3 \right] \not\subseteq 
\left[ i_1 , j_1 \right]$ 
and 
$ \left[ i_3, j_3 \right] \not\subseteq  \left[ \left| p_0  \right|, \left| p \right| - 1 \right].$ 
Furthermore, $\left[ i_3, j_3 \right] \neq \left[ 0, 2 \right]$ due to $\bumps^{\Right} \left( p_0, p_1 \right) = \ze$.
The only possibility is $j_3 = j_1 + 1$ and $ i_3 < j_3$.
This implies that $i_2 = \left| p_0 \right|$, so that $b a^{i_4} \conj b$ is a left factor of $p_1$, where $i_4 := j_2 - i_2 - 1 > 2$; combining this with~\eqref{RefFormula06}, $p_1$ must have the form
$ p_1 = b a p_3 a b^{j_5}$.
Consequently, each of the pairs 
$ \left(  
\left( x_2, y_2 \right) , \iso p \left( \left| p_0 \right| + 2 \right)
\right) $
and
$ \left( 
\left( x_1, y_1 \right) , \iso p \left( \left| p_0 \right| + 1 + \left| p_3 \right| + 1 \right)
\right)$
consists of points having the same abscissas, so that we can invoke~\ref{RefLmGeometric} again to draw that 
$\bumps \left( p_3 \right) \cup \bumps \left( p_6 \right) \neq \ze{}$.
In particular, the set
$
\left\{ \left[ i, j \right] \in
\bumps^{\Up, \Down} \left( p_0, p_1 \right) 
. \left[ i, j \right] \subseteq \left[ i_1 , j_1 \right] \vee \left[ i, j \right] \subseteq \left[ \left| p_0  \right|, \left| p \right| - 1 \right]
 \right\} \sdiff \left\{ \left[ i_2, j_2 \right] \right\}
$
is still not empty, so that we can take an element $\left[ i_6, j_6 \right]$ of it having minimal cardinality.
By an argument similar to that brought forward for $\left[ i_2, j_2 \right]$, we can show that it must be 
$j_6 = \left| p \right| - 1$, $i_6 < j_6$, so that $ \conj b a^{j_4} b$ is a right factor of $p_1$, where $j_4 := j_6 - i_6 - 1$.
Adding~\ref{RefLmStraightPreserved}, this means that $ \left( p_1, p_0 - \left[ 0, 2 \right] \right)$ also belongs to $M$, thus contradicting the minimality of $\left( p_0, p_1 \right)$.
\end{proof}

\begin{Lm}
\label{RefLmClearbumps2}
Let $ \left( q_0, q_1 \right) \in \candidatempty{}$ and $\beta \in \Sigma_2$; assume 
$ \clearbumps^{\beta} \left( q_0, q_1 \right) \cap \argmin_{\size} \bumps^{\beta} \left( q_0, q_1 \right) \neq \ze{}$.
Then there is a $1$-shortcut of $ \left( q_0, q_1 \right)$ which belongs to 
$ \climple \sdiff \trapezoidals{}$. 
\end{Lm}

\begin{proof}
Consider $ \left[ i_0, j_0 \right] \in \clearbumps^{\beta} \left( q_0, q_1 \right) \cap \argmin_{\size} \bumps^{\beta} \left( q_0, q_1 \right)$. 
We define
\begin{align*}
\left( p_0, p_1 \right) := 
\left\{ 
\begin{aligned}
\left( q_0, q_1 \right) & \text{ if } i_0 = 0
\\
\left( \rev{q_1}, \rev{q_0} \right) & \text{ if } j_0 = \left| p_0 p_1 \right| - 1
\\
\left( q_1, q_0 \right) & \text{ if } i_0 = \left| p_0 \right|
\\
\left( \rev{q_0}, \rev{q_1} \right) & \text{ if } j_0 = \left| p_0 \right| - 1.
\end{aligned}
\right.
\end{align*}
Since, by definition of $\clearbumps$, 
$ i_0 = 0 \vee j_0 = \left| p_0 p_1 \right| - 1 
\vee i_0 = \left| p_0 \right| \vee j_0 = \left| p_0 \right| - 1$, we can always obtain at least one $ \left( p_0 , p_1 \right)$ from the definition above.
It is immediate to see that $ \left( p_0, p_1 \right)$ is still in $\candidatempty{}$.
By definition of $ \left( p_0, p_1 \right) $, 
$ \left[ 0, j_1 \right] \in \clearbumps_0^{\alpha} \left( p_0, p_1 \right) \cap \argmin_{\size} \bumps^{\alpha} \left( p_0, p_1 \right)$ for some $\alpha \in \Sigma_2$, $j_1 \in \N$.
We will show that there is a $1$-shortcut of $ \left( p_0, p_1 \right)$ not belonging to $\trapezoidals{}$. 
Since this property is invariant with respect to the operations through which $ \left( p_0, p_1 \right)$ was obtained from $ \left( q_0, q_1 \right)$, this will imply the thesis.
We note that, by hypothesis, $j_1 > 1$; without loss of generality, we can further set $ \alpha = p_0 \left( 0 \right) := b$ and 
$ p_0 \left( 1 \right) := \conj a$.
It must be $p_1 \left( 0 \right) = a$, otherwise $\left( p_0, p_1 \right) - \left\{  0, j_1  \right\}$ would be in $\climple \sdiff \trapezoidals{}$.
Moreover, since $p_0 \left( 0 \right) = b$, it must also be 
$p_0 \left( \left| p_0 \right| - 1 \right) \neq \conj b$.
Hence, $p_0 \left( \left| p_0 \right| - 1 \right) \in \left\{ a, b \right\}$ (because if 
$p_0 \left( \left| p_0 \right| - 1 \right) = \conj a$, simplicity is violated).
Through similar reasoning, we also conclude that 
$p_1 \left( \left| p_1 \right| - 1 \right) = b$.
We now proceed to see that 
\begin{description}
\item[ $p_0 \left( \left| p_0 \right| - 1 \right)$ is not equal to $a$.]
Indeed, if we assume it is, we have the following implications:
\begin{enumerate}
\item
\label{RefImpl1}
by Proposition~\ref{RefLmUpDown}, $\bumps^{\Left} \left( p_0  \right) \neq \emptyset$, which allows us to consider, $ \left[ i_2, j_2 \right] \in \argmin_{\card} \bumps^{\Left} \left( p_0, p_1 \right)$.
\item
\label{RefImpl2}
$ p_0 \left( \left| p_0 \right| - 2 \right) \neq \conj b$.
\item
\label{RefImpl3}
$p_0 \left( 2 \right) \neq \conj b$ (otherwise $p_0$ would be self-factoring).
\end{enumerate}
It is now easy to see that $\doublebumps_1^{\Left}$, $\cornerbumps^{\Left}$, $\cloggedbumps^{\Left}$ and $\clearbumps_0^{\Left}$ are all empty, and from \eqref{RefImpl3} it follows $\doublebumps_0^{\Left}$ also is. 
Moreover, if it were $\left[ i_2, j_2 \right] \in \clearbumps^{\Left}_1$, then $\left( p_0, p_1 \right) - \left\{ i_2, j_2 \right\} $ would be in $\climple{} \sdiff \trapezoidals{}$ (using \eqref{RefImpl2}).
Hence, through~\ref{RefLmBumpClassification4}, we got a contradiction, and must conclude that $p_0 \left( \left| p_0 \right| - 1 \right) = b$.
\end{description}
Let us recapitulate the letters of $p_0$ and $p_1$ we know up to now:
\begin{align}
\label{RefFormula12}
p_0 \left( 0 \right) = p_0 \left( \left| p_0 \right| - 1 \right) = 
p_1 \left( \left| p_1 \right| - 1 \right) = b, 
&& 
p_0 \left( 1 \right) = 
\conj a,
&&
p_1 \left( 0 \right) = a.
\end{align}
Since 
$b {\conj a}^{j_1 - 1} \conj b$ is a left factor of $p_0$, and 
$p_0 \left( \left| p_0 \right|-1 \right) = b$, $\bumps^{\Down} \left( p_0 \right) \neq \emptyset$, 
and we can thus consider 
$ \left( i_3, j_3 \right) \in \argmin_{\size} \bumps^{\Down}$.
It is easy to see that $\doublebumps^{\Down}$, $\cornerbumps^{\Down}_0$, $\cloggedbumps^{\Down}$ are all empty.
Moreover, if it were $\left( i_3, j_3 \right) \in \clearbumps^{\Down}_1$, then it would be $p_0 \left( \left| p_0 \right| - 2 \right) = a$, 
$ j_3 = \left| p_0 \right| - 1$, $i_3 > 0$, 
and consequently $\left( p_0, p_1 \right) - \left\{ i_3, j_3 \right\} \in \climple{} \sdiff \trapezoidals{}$, which is prevented by the definition of $\candidatempty{}$.
Similarly, if it were 
$\left[ i_3, j_3 \right] \in \clearbumps^{\Down}_0 $, it would also be 
$p_1 \left( \left| p_1 \right| - 2 \right) \in \left\{ b, \conj b  \right\}$, 
$ j_3 = \left| p_0 p_1 \right| - 1$, 
$ i_3 > \left| p_0 \right|$:
this would make 
$\left( p_0, p_1 \right) - \left\{ i_3, j_3 \right\}$ a member of $ \climple{} \sdiff \trapezoidals{}$, 
against the definition of $\candidatempty{}$.
$ \left[ i_3, j_3 \right]$
must then belong to $\cornerbumps^{\Down}_1$ by~\ref{RefLmBumpClassification4};
this implies that $\conj b \conj a^{j_3 - i_3 - 1} b$ is a right factor of $p_0$.
Moreover, $ b \conj a^{j_0 - 1} \conj b$ is a left factor of $p_0$, because 
$ \left[ 0, j_0 \right] \in \clearbumps^{b} \left( p_0, p_1 \right)$ and $p_0 \left( 1 \right) = \conj a$.
Invoking Lemma~\ref{RefLmTopological}, we draw that 
\begin{align}
\label{RefFormula10}
\bumps^{a, \conj a} \left( p_0, p_1 \right) \neq \ze{}.
\end{align}
We preliminarily observe that
\begin{align}
\label{RefFormula11}
 \clearbumps_0^{a, \conj a} =  
 \clearbumps_1^{\Left} =
 \doublebumps_0^{\Right} =  
 \doublebumps_1^{a} =  
 \doublebumps_1^{\Left} = 
 \cornerbumps_0^{a, \conj a} = 
 \cornerbumps_1^{a} = 
 \cornerbumps_1^{\Left} = 
 \cloggedbumps_0^{\Left} =
 \cloggedbumps_1^{\Right} = \ze{},
\end{align}
where
 $\doublebumps_0^{\Right} = \ze{}$ 
follows from $p_1$ not being self-factoring, 
$\doublebumps_1^{a} = \ze{} $ 
is due to the simplicity of $p_0 p_1$, and the remaining facts are immediate consequence of the respective definitions and of~\eqref{RefFormula12}.
Let us finally proceed by cases using~\eqref{RefFormula10}.
\begin{description}
\item[Case $\ze{} \neq \argmin_{\size} \bumps^{\Right} \left( p_0, p_1 \right)$: ]
consider $\left[ i_4, j_4 \right]$ in it.
If it were $ \left[ i_4, j_4 \right] \in \clearbumps_1^{\Right}$, then it would be 
$ p_1 \left( 1 \right) = b$ due to~\eqref{RefFormula12} and simplicity.
This would mean $\left( p_0, p_1 \right) - \left\{ i_4, j_4 \right\} \in \climple \sdiff \trapezoidals{}$, against the definition of $ \candidatempty{}$.
The only possibility left from~\eqref{RefFormula11} and~\ref{RefLmBumpClassification4} is then
$ \left[ i_4, j_4 \right] \in \cloggedbumps_0^{\Right}$, and thesis follows immediately from~\ref{RefLmCloggedShortcut} and~\eqref{RefFormula12}.
\item[Case $ \ze{} \neq \argmin_{\size} \bumps^{\Left} \left( p_0, p_1 \right)$: ]
consider $ \left[ i_5, j_5 \right]$ in it, and suppose 
$ \left[ i_5, j_5 \right] \in \doublebumps_0^{\Left}$.
This implies that $b \conj a \conj b^{k_0} a$ is a left factor of $p_0$ for some $k_0 \geq 3$.
Since $p_1 \left( 0 \right) = a $ and $ p_1 \left( \left| p_1 \right| - 1 \right) = b$, $\gamma b^{k_1}$ must be a right factor of $p_1$ for some $\gamma \neq b$, $k_1 \in \Z^+$.
Due to simplicity, $k_1 \leq k_0 - 2$; in turn, this implies that $\gamma$ cannot be $a$, otherwise simplicity would be again violated.
Ultimately, $\conj a b^{k_1}$ must be a right factor of $p_1$; via~\ref{RefLmUpDown}, this implies $ \bumps^{\Right} \left( p_1 \right) \neq \ze{}$.
In particular, we can trigger the previous case and get the thesis.
Hence, the only possibility left from~\eqref{RefFormula11} and~\ref{RefLmBumpClassification4} is 
$ \left[ i_5, j_5 \right] \in \cloggedbumps_1^{\Left}$.
This implies (Proposition~\ref{RefLmSymmetry}) that $\cloggedbumps_0^{\Left} \left( p_1, p_0 \right) \neq \ze{}$, so that there must be (Proposition~\ref{RefLmCloggedShortcut}) either a $1$- shortcut of $p_1$ having the form $(0, l, \conj a)$ or a $1$-shortcut of $p_0$ having the form 
$ \left( l, \left| p_0 \right| - 1, a \right)$.
In the first case, the first letter of the obtained shortcut of $p_1$ becomes $\conj a$, while the first letter of $p_0$ remains $b$, and therefore the obtained pair is still not in $\trapezoidals{}$.
Similarly, in the second case, the last letter of the obtained shortcut of $p_0$ is $a$, while the last letter of $p_1$ remains $b$, and hence the obtained shortcut pair is still not in $\trapezoidals{}$.
\end{description}
\end{proof}

\begin{Cor}
\label{RefLmClearDoubleBumps}
\sloppy{} 
Assume 
$( \clearbumps^{\beta} \left( p_0, p_1 \right)$ $\cup$
$ \doublebumps^{\beta} \left( p_0 ,p_1 \right))$ $\cap$
$\argmin_{\size} \bumps^{\beta} \left( p_0, p_1 \right) 
\neq \ze{}$. 
for some $\left( p_0, p_1 \right) $$ \in$$ \candidatempty{}$ and $\beta $$\in$$ \Sigma_2$.
Then there is a $1$-shortcut of $ \left( p_0, p_1 \right)$ which belongs to 
$ \climple \sdiff \trapezoidals{}$.
\fussy{}
\end{Cor}

\begin{proof}
Consider 
$\left[ i, j \right] \in$
$( \clearbumps^{\beta} \left( p_0, p_1 \right)$ $ \cup$  
$\doublebumps^{\beta} \left( p_0 ,p_1 \right) )$
$\cap \argmin_{\size} \bumps^{\beta} \left( p_0, p_1 \right).$
If $\left[ i, j \right] \in \clearbumps^{\beta} \left( p_0, p_1 \right) 
\cap \argmin_{\size} \bumps^{\beta} \left( p_0, p_1 \right),$
we apply Lemma~\ref{RefLmClearbumps2}.
Otherwise, 
$\left[ i, j \right] \in
\doublebumps^{\alpha} \left( p_0 ,p_1 \right)
\cap \argmin_{\size} \bumps^{\alpha} \left( p_0, p_1 \right)$
for some $\alpha \ortho \beta$, and we can still apply Lemma~\ref{RefLmClearbumps2}. 
\end{proof}


\begin{Lm}
\label{RefLmCornerBumps}
Given $ \left( p_0, p_1 \right) \in \candidatempty{},$
assume that $\cornerbumps^{\beta} \left( p_0, p_1 \right) \neq \ze{}.$
Then there is a $1$-shortcut of $ \left( p_0, p_1 \right)$ which belongs to 
$ \climple \sdiff \trapezoidals{}$. 
\end{Lm}

\begin{proof}
We can assume $\beta = b$ and, by Proposition~\ref{RefLmSymmetry}, 
$\cornerbumps^b_0 \left( p_0, p_1 \right) \neq \ze{}$.
Consider then $\left[ i_0, j_0+1 \right] \in \cornerbumps^b_0 \left( p_0, p_1 \right)$.
By possibly employing the transformation 
$
\left( p_0, p_1 \right) \mapsto
\left( 
\rev { \conj {p_1} }, \rev {\conj {p_0} } 
\right),
$
we can further impose that $i_0 = 0$, and finally $p_0 \left( 1 \right) =\conj a$ by possibly swapping $a$'s with $\conj a$'s.
We therefore must conclude that $p_0$ has the form $p_0 = b \conj a ^ j_0 \conj b p_1$ for some $p_1 \in \letters_2^*$, and that $p_1$'s last letter is $a$.
Now, it cannot be $p_1 \left( 0 \right) = \conj a$, otherwise $p_1$ would be self-factoring.
Moreover, it cannot be $p_1 \left( 0 \right) = \conj b$, which would make 
$\left( p_0 , p_1 \right) \in \trapezoidals{}$.
Proceeding with similar reasoning, we obtain $
\left\{ 
p_0 \left( \left| p_0 \right| -1 \right)
, p_1 \left( 0 \right)
 \right\} \subseteq 
\left\{ a, b  \right\}$.
\begin{description}
\item
[Case $p_0 \left( \left| p_0 \right| - 1 \right) = a$]
Then $\bumps^{\Left}\left( p_0 \right) \neq \ze{}$ by Proposition~\ref{RefLmUpDown}, and we can consider 
$\left[ i_1, j_1 \right] \in$ \\ $\argmin_{\card} \bumps^{\Left} \left( p_0, p_1 \right)$.
If $\left[ i_1, j_1 \right] \in \clearbumps^{\Left}\left( p_0, p_1 \right) 
\cup \doublebumps^{\Left} \left( p_0, p_1 \right)$, then we invoke
Corollary~\ref{RefLmClearDoubleBumps}.
$\cornerbumps_0^{\Left} \left( p_0, p_1 \right) = \ze$ by simplicity, and
$\cloggedbumps^{\Left} \left( p_0, p_1 \right) = \ze$,
hence it only remains to check the case $ \left[ i_1, j_1 \right] \in \cornerbumps_1^{\Left} \left( p_0, p_1 \right)$, 
which implies $p_1 \left( 0 \right) = b$ and $p_0 \left( \left| p_0 \right| - 2 \right) = \conj b$.
But the latter would imply that $p_0$ is self factoring.
\item[Case $p_0 \left( \left| p_0 \right| - 1 \right) = b$]
This implies that $\bumps^{\Down} \left( p_0 \right) \neq \ze$ via Proposition~\ref{RefLmUpDown}, so that we can consider 
$ \left[ i_1, j_1 \right] \in \argmin_{\card} \bumps^{\Down} \left( p_0, p_1 \right)$.
If $\left[ i_1, j_1 \right] \in \clearbumps^{\Down} \left( p_0, p_1 \right) \cup \doublebumps^{\Down} \left( p_0, p_1 \right)$, 
then we invoke Corollary~\ref{RefLmClearDoubleBumps}.
Moreover, $\cornerbumps_0^{\Down} \left( p_0, p_1 \right) = \ze$, and $\cloggedbumps_1^{\Down} \left( p_0, p_1 \right) = \ze$.
We also note that if $\cloggedbumps_0^{\Down} \left( p_0, p_1 \right) \neq \ze$, we would violate simplicity.
Ultimately, then, we only need to check the case $\left[ i_1, j_1 \right] \in \cornerbumps_1^{\Down} \left( p_0, p_1 \right)$, 
which implies $p_1 \left( 0 \right) = a$ and $p_0=p_2 \conj b \conj a^k b$ for some $k \in \Z^+$, $p_2 \in \letters_2^*$;
as a consequence, $\bumps^{\left\{ a, \conj a \right\}} \left( p_0 \right) \neq \ze$ via 
Lemma~\ref{RefLmTopological}, and we can consider $ \left[ i_2, j_2 \right] \in 
\argmin_{\size} \bumps^{\left\{ a, \conj a \right\}} \left( p_0, p_1 \right)$.
If 
$ \left[ i_2, j_2 \right] \in 
\left( 
\clearbumps^{ \left\{ a, \conj a \right\}} \left( p_0, p_1 \right)
\cup
\doublebumps^{ \left\{ a, \conj a \right\} } \left( p_0, p_1 \right) 
\right),
$
then we apply Corollary~\ref{RefLmClearDoubleBumps}.
Otherwise, we note that $\cornerbumps^{\left\{ a, \conj a \right\}} \left( p_0, p_1 \right) = \ze$, so that it must be $ \left[ i_2, j_2 \right] \in 
\cloggedbumps^{\left\{ a, \conj a \right\}} \left( p_0, p_1 \right)
=
\cloggedbumps^a_0 \left( p_0, p_1 \right) \cup 
\cloggedbumps^{\conj a}_1 \left( p_0, p_1 \right),
$
If $\left[ i_2, j_2 \right] \in 
\cloggedbumps^a_0 \left( p_0, p_1 \right)$, we apply Proposition~\ref{RefLmCloggedShortcut} and check that the obtained shortcut has the wanted property.
Otherwise, if 
$ \left[ i_2, j_2 \right] \in  
\cloggedbumps^{\conj a}_1 \left( p_0, p_1 \right),
$
we consider $ \left( p_1, p_0 \right)$, apply Proposition~\ref{RefLmSymmetry} and again 
Proposition~\ref{RefLmCloggedShortcut}, finally checking
that the obtained shortcut has the wanted property.
\end{description}
\end{proof}

\subsection{Bumps in $\cloggedbumps$ and final proof}

\begin{Lm}
\label{RefLmCloggedBumps}
Given $ \left( p_0, p_1 \right) \in \candidatempty$, assume 
$ \cloggedbumps \left( p_0, p_1 \right) \neq \ze$.
Then there is a $1$-shortcut of $ \left( p_0, p_1 \right)$ which belongs to 
$ \climple \sdiff \trapezoidals$. 
\end{Lm}

\begin{proof}
Set $p := p_0 p_1$.
We can assume $p_0 \left( 0 \right) = a$ and the existence of $i$, $j$ such that
$ \left[ i, j \right] \in \cloggedbumps_0^{\Up} \left( p_0, p_1 \right)$.
Now, by simplicity and by the definition of $\cloggedbumps_0$, it must be $ \rev{p_1} \left( 0 \right) \in \left\{ a, b \right\} $.
Additionally, by definition of $\bumps^{\Up} \left( p_0, p_1 \right)$, it must be
either $j < \left| p_0 \right|$ or $i \ge \left| p_0 \right|$.
\begin{description}
\item[Case $ j < \left| p_0 \right|$: ]
then there is a unique $k$ such that $\iso {p_0} \left( k \right) = \left( 1,1  \right)$, and $k$ must be in $ \left]i , j  \right[ $;
it is immediate to check that 
$ \left( 2, k-1, b \right) $ is a $1$-shortcut of $p_0$.
Since this shortcut does not change the first and last letters of $p_0$, we obtain the thesis.
\item[Case $i \ge \left| p_0 \right|$ and $\rev{p_1} \left( 0 \right) = a$: ]
similar to the previous case.
\item[Case $i \ge \left| p_0 \right|$, $\rev{p_1} \left( 0 \right) = b$ and 
$ \rev{p_0} \left( 0 \right) \neq b$: ]
applying Proposition~\ref{RefLmCloggedShortcut}, we obtain a shortcut of $p_1$ changing the last letter of $p_1$ into $\conj b$. 
This allows thesis since the last letter of $p_0$ is not $b$.
\item[Case $i \ge \left| p_0 \right|$, $\rev{p_1} \left( 0 \right) = b$ and 
$ \rev{p_0} \left( 0 \right) = b$: ]
then the sequence $\iso p$ reaches the point $\left( 0, 1 \right) $ at some index $k \in \left[ \left| p_0 \right|, \left| p \right| - 2 \right]$, while
$ p \left( \left| p \right| - 1 \right) = \left( 0, -1 \right)$ and 
$ p \left( \left| p \right| \right) = \left( 0, 0 \right)$.
By Proposition~\ref{RefLmUpDown}, then, $\bumps^{\Down} \left( p_1 \right) \neq \ze$.
But $\rev{p_1} \left( 0 \right) = b$ implies $\cloggedbumps_0^{\Down} \left( p_0, p_1 \right) = \ze$, and $\rev{p_0} \left( 0 \right) = b$ implies $\cloggedbumps_1^{\Down} \left( p_0, p_1 \right)= \ze$.
Hence it must be $\clearbumps^{\Down} \left( p_0, p_1 \right) \cup \doublebumps^{\Down} \left( p_0, p_1 \right) \cup \cornerbumps^{\Down} \left( p_0, p_1 \right) \neq \ze$ by Lemma~\ref{RefLmBumpClassification4}, and we can use
Corollary~\ref{RefLmClearDoubleBumps} or Lemma~\ref{RefLmCornerBumps}.
\end{description}
\end{proof}

\begin{proof}[Proof of Theorem~\ref{RefLmMain1}]
Using hypotheses, $ \left( p_0, p_1 \right) \in \candidatempty{}$ by definition of $\candidatempty{}$.
Using hypothesis~\eqref{RefAss1} and Lemma~\ref{RefLmBumpClassification4}, we have 
$ \left( 
\argmin_{\size} \bumps^{\alpha} \left( p_0, p_1 \right)  \right)
\cap
\left(
\clearbumps^{\alpha} \left( p_0, p_1 \right) 
\cup \doublebumps^{\alpha}  \left( p_0, p_1 \right)
\cup \cornerbumps^{\alpha}  \left( p_0, p_1 \right)
\cup \cloggedbumps^{\alpha} \left( p_0, p_1 \right)
\right) \neq \ze
$
for some $\alpha \in \letters_2$.
If $ \left( \argmin_{\size} \bumps^{\alpha} \left( p_0, p_1 \right)  \right)
\cap
\left(
\clearbumps^{\alpha} \left( p_0, p_1 \right) 
\cup \doublebumps^{\alpha}  \left( p_0, p_1 \right)
\right) \neq \ze,$
we apply Corollary~\ref{RefLmClearDoubleBumps}, if 
$\left( \argmin_{\size} \bumps^{\alpha} \left( p_0, p_1 \right)  \right)
\cap \cornerbumps^{\alpha}  \left( p_0, p_1 \right) \neq \ze$
we apply Lemma~\ref{RefLmCornerBumps}, and if 
$ \left( \argmin_{\size} \bumps^{\alpha} \left( p_0, p_1 \right)  \right)
\cap \cloggedbumps^{\alpha}  \left( p_0, p_1 \right) \neq \ze$ 
we apply Lemma~\ref{RefLmCloggedBumps}.
\end{proof}


\section{Proof of Theorem~\ref{RefLmMain2}}
\label{RefSectLmMain2}

\begin{Prop}
\label{RefLmSelfFact1}
Let $p_0, q_1, q_2 \in \Sigma_2^*$, $m_1, m_2 \in \N$, $\alpha, \beta \in \Sigma_2$, 
$p := p_0 \beta^{m_1+m_2} q_1 q_2 \in \Sigma_2^*$, and assume
\begin{enumerate}
\item
\label{RefAssFact2}
$p_0 \beta^{m_1} q_2 \in \left[ \ze \right]_{\sim} \sdiff \left\{ \ze \right\}.$
\setcounter{propertyCounter}{\value{enumi}}
\end{enumerate}
Moreover, suppose that, if 
$q_1 \neq \ze$ and $\min \left\{ m_1, m_2 \right\} > 0$,
then the following requirements all hold:
\begin{enumerate}
\setcounter{enumi}{\value{propertyCounter}}
\item
\label{RefAssOrtho2}
$\alpha \ortho \beta$;
\item
\label{RefAssNoInversions}
No factor of $q_1 q_2$ belongs to 
$\left\{ \beta \conj \beta, \conj \beta \beta \right\}$;
\item
\label{RefAssSimple}
no non-empty, proper left factor of $p_0 \beta^{m_1 + m_2}$ belongs to $\left[ \ze \right]_\sim$;
\item
\label{RefAssMinimalBump}
$\forall \left[ i, j \right] \in \bumps^{\alpha} \left( q_1 q_2 \right).\ 
j-i \ge m_1+m_2
$;
\item
\label{RefAssNoTrivialFactors}
If $p_0 = \ze$ and $\conj \beta$ is a right factor of $p$, then $\conj{\beta}^n$ is a right factor of $p$ for some $n \ge m_1 + m_2$.
\end{enumerate}
Then $p_0 \alpha \beta^{m_1 + m_2} \conj \alpha q_1 q_2$ is self-factoring.
\end{Prop}

\begin{proof}
Immediate for 
$\min \left\{ m_1, m_2 \right\} = 0 \vee q_1 = \ze$, hence assume $m_1, m_2 > 0$ and $q_1 \neq \ze$.
Then, using hypothesis~\eqref{RefAssFact2}, the following set is non-empty:
\begin{align*}
\argmin_{\size} \left\{ \left[ i, j \right] \subset \dom p. 
\factor p {i} {j-1} \equiv_2 p
\wedge i \in \left] \left| p_0 \right|, \left| p_0 \right| + m_1 + m_2 \right[
\wedge j \ge \left| p_0 \right| + m_1 + m_2
\right\};
\end{align*}
therefore, we can consider $\left[ i, j \right]$ in it such that $i$ is minimal.
Setting $n_1 := i - \left| p_0 \right| > 0$, $n_2 := m_1 + m_2 - n_1$, 
and $p_2 := \factor p j {\left| p \right| - 1}$,
we have that 
$p = p_0 \beta^{n_1+n_2} p_1 p_2$ for some $p_1$, and 
$ p_0 \beta^{n_1} p_2 \in \left[ \ze \right]_{\eq} \sdiff \left\{ \ze \right\}$ by construction.
If $n_2 = 0$ or $p_1 =\ze$, thesis is immediate, hence assume $n_2>0$ and $p_1 \neq \ze$.
From $p_0 \beta^{n_1}$ not being closed by hypothesis~\eqref{RefAssSimple}, we can that $p_2$ is neither.
By construction of $\left[ i, j \right]$, $\rev {p_1} \left( 0 \right) \ortho \beta$; 
moreover, if 
$ \rev{p_1} \left( 0 \right) = \conj{\alpha},
$ thesis is again immediate, so that we also assume $p_1 = p_1' \alpha$ for some $p_1'$.
Furthermore, thanks to hypotheses~\ref{RefAssNoInversions} and \ref{RefAssOrtho2}, we can write 
$p_2 = \beta_1^{k_1} p_2'$ with $\beta_1 // \beta$, 
$p_2'$ having no left factor in $\left\{ \beta, \conj \beta \right\}$, 
and $k_1 \in \N$.
\begin{description}
\item[Case $\beta_1=\beta$: ]
We have
$p = 
p_0 \beta^{n_1 + l} 
\beta^{n_2 - l} p_1' \alpha \beta^l 
\beta^{k_1 - l} p_2'
\text{ and }
p_0 \beta^{n_1 + l} 
\beta^{k_1 - l} p_2' \eq \ze,$
with $l $ $:= $ $\min $ $\{ n_2, k_1 \}$.
If $k_1 < n_2$, this yields $ p_2' \notin \left[ \ze \right]_{\eq}$ by hypothesis~\eqref{RefAssSimple}, so that $ 
p_2' \left( 0 \right)
$ cannot be $\conj \alpha$ by hypothesis~\eqref{RefAssMinimalBump}, and therefore it must be 
$ 
p_0 \beta^{n_1 + l} 
\alpha \factor {p_2'}{1}{\left| p_2' \right|-1}
\in \left[ \ze \right] \sdiff \left\{ \ze \right\}
\ra{}
$
$
p_0 \alpha \beta^{n_1 + l} 
\factor {p_2'}{1}{\left| p_2' \right|-1}
\in \left[ \ze \right] \sdiff \left\{ \ze \right\},
$
yielding the thesis.
If, on the other hand, $k_1 \ge n_2$, then
$ 
p_0 \beta^{n_1 + n_2}
\beta^{k_1 - n_2} p_2' \in \left[ \ze \right] \sdiff \left\{ \ze \right\}
\ra{}
$
$
p_0 \alpha \beta^{n_1 + n_2} \conj{\alpha}
\beta^{k_1 - n_2} p_2' \in \left[ \ze \right] \sdiff \left\{ \ze \right\}.
$

\item[Case $\beta_1=\conj \beta$: ]
We have
\begin{align}
\label{RefFormula14}
p = p_0 \beta^{n_1 - l} \beta^{n_2 + l} p_1' \alpha \conj{\beta}^l 
\conj{\beta}^{k_1 - l} p_2'
&& \text{ and } &&
p_0 \beta^{n_1 - l}  
\conj{\beta}^{k_1 - l} p_2' \in 
\left[ \ze \right],
\end{align}
where
$l := \min \left\{ n_1, k_1  \right\}$.
If $n_1 > k_1$, hypothesis~\eqref{RefAssSimple} implies that $p_2' \notin \left[ \ze \right]_{\eq}$, so that we can consider 
$ p_2' \left( 0 \right)$;
the latter cannot be $\conj \alpha$ due to hypothesis~\eqref{RefAssMinimalBump}, 
therefore $p_2' \left( 0 \right) = \alpha$, giving thesis immediately.
If $n_1 = k_1 > 0$, then
$
p_0   
p_2' \in 
\left[ \ze \right],
$
and $p_0 p_2' \neq \ze$ by hypothesis~\eqref{RefAssNoTrivialFactors}.
Finally, if $n_1 < k_1$, then \eqref{RefFormula14} gives $p_0 \conj{\beta}^{k_1 - l}p_2' \in \left[ \ze \right]_{\eq} \sdiff \left\{ \ze \right\}$.
\end{description}
\end{proof}

\begin{Lm}
\label{RefLmSelfFact2}
Given $\alpha \in \Sigma_2$, $q \in \Sigma_2^*$, $ \left[ i_0, j_0 \right] \in \argmin_{\size} \bumps^{\alpha} \left( q \right)$, assume that, if $\left| q \right| > 2$, then all the following hypotheses hold:
\begin{enumerate}
\item
\label{RefAssSelfFact}
$ q - \left\{ i_0, j_0 \right\}$ is simple and self-factoring;
\item
\label{RefAssNonCritical}
$
\left[ i_0, j_0 \right] \notin \criticalbumps \left( q \right),
$
where, for any $p$, $\criticalbumps \left( p \right)$ is defined as
\begin{align*}
\criticalbumps \left( p \right) :=
\{  
\left[ i, j \right] \in \bumps \left( p \right). \ 
\left( i = 0 \wedge j - i > 2 \wedge j < \left| p \right| - 1 \wedge p \left( \left| p \right| - 1 \right) = \conj{p \left( 1 \right)}
 \right)
\vee
\\
\left( j = \left| p \right| - 1 \wedge j - i > 2 \wedge i > 0 \wedge p \left( 0  \right) = \conj{p \left( \left| p \right| - 2 \right)}
 \right)
\}.
\end{align*}
\end{enumerate}
Then $q$ is self-factoring.
\end{Lm}

\begin{proof}
We can assume $\left| q \right| > 2$, 
$ j_0 - i_0 > 2 $, 
$ \left[ i_0, j_0 \right] \subset \left[ 0, \left| q \right| - 1 \right]$, and
$ q \notin \left[ \ze \right]_{\sim}$.
Set $p := q - \left\{ i_0, j_0 \right\}$, 
thereby having $ p = \factor q 0 {i_0-1} \ 
\beta^{j_0 - i_0 - 1} \ 
\factor q {j_0+1} {\left| q \right| - 1}
$
for some $\beta \ortho \alpha$ 
and consider, by hypothesis~\ref{RefAssSelfFact}, $i_1 \le j_1$ satisfying
$ \factor p 0 {i_1} \ \factor p {j_1} {\left| p \right| - 1} \in \left[ \ze \right] \sdiff \left\{ \ze \right\}$. 
Now, if $ \left| \left[ i_0-1, j_0-1 \right] \cap \left\{ i_1, j_1 \right\} \right|$ is even, the thesis is immediate; 
otherwise, we have two cases:
\begin{description*}
\item[Case $i_1 \in {\left[ i_0-1, j_0-2 \right]}$ and $ j_1 \ge j_0$: ] 
then\\
$
p = \factor p 0 {i_0 - 1} \ 
\beta^{i_1 - i_0 + 1} \  \beta^{j_0 - i_1 - 2} \  
\factor p {j_0-1} {j_1-1} \ 
\factor p {j_1} {\left| p \right| - 1} 
\  \text{ and } \   
\factor p 0 {i_0 - 1}  
\beta^{i_1 - i_0 + 1} \ 
\factor p {j_1} {\left| p \right| - 1} \in 
\left[ \ze \right]_{\eq} \sdiff \left\{ \ze \right\},
$\\
so that we can apply Proposition~\ref{RefLmSelfFact1} to obtain the thesis.
Note that, if 
$ \factor p 0 {i_0 - 1} = \ze$, then $q \left( 1 \right)=\beta$, and therefore
$ \rev p \left( 0 \right) = \rev q \left( 0 \right) \neq \conj \beta$ because $\left[ i_0, j_0 \right] \notin \criticalbumps \left( q \right)$, allowing to satisfy hypothesis~\ref{RefAssNoTrivialFactors} of Proposition~\ref{RefLmSelfFact1}.
\item[Case $j_1 \in {\left[ i_0-1, j_0-2 \right]}$ and $ i_1 < i_0-1$: ]
we consider $\rev p$ and reason similarly.
\end{description*}
\end{proof}

\begin{Prop}
\label{RefLmNonCriticalBumps}
Let $ \left( p_0, p_1 \right) \in \climple \sdiff \trapezoidals \sdiff \selffact{}$, 
$ \left[ i_0, j_0 \right] \in \bumps^{\alpha} \left( p_0 \right) \cap \criticalbumps \left( p_0 \right)$,
where $ \criticalbumps$ is defined as in Lemma~\ref{RefLmSelfFact2}.
Assume 
\begin{enumerate*}
\item
\label{RefAssShortBumps}
$ \forall i. \left[ i, i+2 \right] \in \bumps \left( p_0, p_1 \right) 
\ra{} \left( p_0, p_1 \right) - \left\{ i, i+2 \right\} \notin \climple \sdiff \trapezoidals$;
\item
$ \left( p_0, p_1 \right) - \left\{ i_0, j_0 \right\} \in \climple \sdiff \trapezoidals$.
\end{enumerate*}
Then there are $\beta \ortho \alpha$ and $ \left[ i_1, j_1 \right] \in 
\argmin_{\size} \bumps^{\beta} \left( p_0, p_1 \right) 
\sdiff \criticalbumps \left( p_0 \right)
\sdiff \left( \left\{ \left| p_0 \right|  \right\} + \criticalbumps \left( p_1 \right) \right)
$
such that 
\\
$ \left( p_0 - \left\{ i_1, j_1 \right\}, p_1 \right) \in \climple \sdiff \trapezoidals$.
\end{Prop}

\begin{proof}
Looking at the definition of $\criticalbumps$, we can assume 
\begin{align}
\label{RefFormula15}
i_0=0, && j_0 < \left| p_0 \right| - 1, 
&& p_0 \left( 0 \right) = b, 
&& p_0 \left( 1 \right) = p_0 \left( 2 \right) = a, &&
&& p_0 \left( \left| p_0 \right| - 1 \right) = \conj a.
\end{align}
Since $ \left( p_0, p_1 \right) \in \climple \sdiff \trapezoidals \sdiff \selffact$ and 
$ \left( p_0, p_1 \right) - \left\{ i_0, j_0 \right\} \in \climple \sdiff \trapezoidals$, 
it also must be 
\begin{align}
\label{RefFormula16}
\rev {p_1} \left( 0 \right) = b = p_1 \left( 0 \right) &&  \text{ and } &&
p_0 \left( \left| p_0 \right| - 2 \right) \neq \conj b.
\end{align}
Therefore, we can use simplicity of $ p_0 p_1$ and the respective definitions to check that
\begin{align}
\label{RefFormula17}
\clearbumps_0^{\Right} \left( p_0, p_1 \right) =
\doublebumps^{\Right} \left( p_0, p_1 \right) 
= 
\cornerbumps^{\Right} \left( p_0, p_1 \right) 
=
\cloggedbumps^{\Right} \left( p_0, p_1 \right) 
= \ze.
\end{align}
Combining \eqref{RefFormula15} and \eqref{RefFormula16} with hypothesis~\ref{RefAssShortBumps}, one can check that if there is a bump of $ \bumps \left( p_0, p_1 \right)$ having cardinality $3$, it must belong to the set
$ \left\{ \left[ \left| p_0 p_1 \right| - 3, \left| p_0 p_1 \right| - 1 \right] \right\}$.
By virtue of \eqref{RefFormula15} and Proposition~\ref{RefLmUpDown}, we can consider 
$ 
\left[ i_1, j_1 \right] \in \argmin_{\size} \bumps^{\Right} \left( p_0, p_1 \right)
$ and, by what we said in the last sentence, assert that
$ 
\left\{ \left[ i_1 - 1, i_1 + 1 \right], \left[ j_1 - 1, j_1 + 1 \right] \right\}
\cap
\bumps^{a^{\ortho}} \left( p_0, p_1 \right) = \ze.
$.
Adding~\eqref{RefFormula17} and Proposition~\ref{RefLmBumpClassification3},
if it were $\left( p_0, p_1  \right) - \left\{ i_1, j_1 \right\} \notin \climple \sdiff \trapezoidals$, then we could conclude 
$\left[ i_1, j_1 \right] \in \clearbumps^{\Right}_1 \left( p_0, p_1 \right)$ ($\left[ i_1, j_1 \right]$ cannot nest any other bump of $ \bumps^{\Right} \left( p_0, p_1 \right)$ being minimal).
By definition of $\clearbumps_1$, this implies $j_1 = \left| p_0 \right| - 1$, 
$ \left[ i_1, j_1 \right] \in \bumps^{\Right} \left( p_0 \right)$ and 
$ \left( p_0, p_1 \right) - \left\{ i_1, j_1 \right\} = 
\left( p_0 - \left\{ i_1, j_1 \right\}, p_1 \right) \in \climple.
$
It is immediate to check that $ \left( p_0 - \left\{ i_1, j_1 \right\}, p_1 \right) $ is not in 
$\trapezoidals$, so that we have a contradiction and it must be 
$ \left( p_0, p_1 \right) - \left\{ i_1, j_1 \right\} \in \climple \sdiff \trapezoidals$.
By definition of $\criticalbumps$, \eqref{RefFormula15} and \eqref{RefFormula16}, we can check that
$ \criticalbumps\left( p_1 \right) \cap \bumps^{\Right} \left( p_1 \right) = \ze 
= \criticalbumps \left( p_0 \right) \cap \bumps^{\Right} \left( p_0 \right).
$
\end{proof}

\begin{proof}[Proof of Theorem~\ref{RefLmMain2}]
By simplicity, any bump for $ \left( \ceP, \ceQ \right)$ must have cardinality at least $3$.
If there were a minimal bump for $ \left( \ceP, \ceQ \right)$ of the form $\left[ i, i+2 \right]$ and such that $ \left( \ceP, \ceQ \right) - \left\{ i, i+2 \right\} \in \climple \sdiff \trapezoidals$, then we 
would have that $ \left( \ceP, \ceQ \right) - \left\{ i, i+2 \right\} \in \selffact$ due to~\ref{RefReqMin}; 
let us assume $ \left[ i, i+2 \right] \in \bumps \left( \ceP \right)$: $ p_0 - \left\{ i, i+2 \right\} $ is then simple and self-factoring, while $ \left[ i, i+2 \right]$ cannot be in $\criticalbumps \left( \ceP \right)$ being too short.
Therefore, we could invoke Lemma~\ref{RefLmSelfFact2} to establish that $\ceP$ is self-factoring, while $ \left( \ceP, \ceQ \right) \notin \selffact$.
We have thus proved that $ \left( \ceP, \ceQ \right)$ satisfies hypothesis~\ref{RefAssShortBumps} of Proposition~\ref{RefLmNonCriticalBumps}, and can consequently prove that, if there is any bump 
$ \left[ i_0, j_0 \right]$ such that $ \left( \ceP, \ceQ \right) - \left\{ i_0, j_0 \right\}
\in \climple \sdiff \trapezoidals$, 
then there is a minimal bump $ \left[ i_1, j_1 \right]$ for $ \left( \ceP, \ceQ \right)$ 
not belonging to $ \criticalbumps \left( \ceP \right) \cup 
\left( \left\{ \left| \ceP \right| \right\} + \criticalbumps \left( \ceQ \right) \right)$ 
and such that 
$ \left( \ceP, \ceQ \right) - \left\{ i_1, j_1 \right\} \in \climple \sdiff \trapezoidals.$
By~\ref{RefReqMin}, then, we have $ \left( \ceP, \ceQ \right) - \left\{ i_1, j_1 \right\} \in \selffact$.
As before, we can assume $ \left[ i_1, j_1 \right] \in \bumps \left( \ceP \right)$, which means
$ \ceP - \left\{ i_1, j_1 \right\}$ is simple and self-factoring.
Applying Lemma~\ref{RefLmSelfFact2} again, we conclude that 
$\ceP$ is also self factoring, which is impossible because 
$ \left( \ceP, \ceQ \right) \notin \selffact$.
\end{proof}

\section{Conclusions}
\label{RefSectConcl}
The results linking $O_N$ and MCFGs provided a starting point for the main question of this paper: that is, which alterations can be made to the notion of a MCFG without severing these links.
We provided a first answer, by slightly tweaking the definition of a MCFG and showing that the link persists, at least for the casen $N=2$. 
We cursorily note that Theorem~\ref{RefLmMain} still remains valid if the conditions on rule~\eqref{RefRuleRectangular} are moved to rule~\eqref{RefRuleSandwich}, by symmetry.
Besides the proof-theoretical interest of a result of this type, it is hoped that the direction pointed out here can be helpful towards establishing further links between $O_N$ and structures that are independent of MCFGs, such as indexed grammars, for which the study of such links still poses a number of open problems~\cite{salvati2015mix}.
A more immediate application of the results presented here is a computational one: we have established that any binary factorisation of a closed, simple walk in $\Z^2$ is either in $\trapezoidals$ or $\selffact$. 
As the $\trapezoidals$-membership problem is trivial to decide, and since the notion of a self-factoring string is intimately connected to the word problem in $\Z^2$, the present results can be potentially used for arbitrary binary factorisations of closed, simple walks in relation to algorithmic aspects of that problem. 
Related to this, the style of the proofs provided here has deliberately been kept close to low-level, set-theoretical concepts and free of geometrical arguments (similarly to what done in~\cite{caminati2024Dlt}), so to make at least conceivable to work on a future formalisation in a proof assistant for a verification of the presented results and of possible algorithms deriving from them.

\textcolor{white}{{\fontsize{1}{1}\selectfont 
This preprint has
not undergone 
any post-submission improvements or
corrections. 
The Version of Record of this contribution is published in MCU 2026,
and is available online at 
https://doi.org/10.1007/978-3-031-81202-6
}}

\subsubsection*{Acknowledgments}
The author is grateful to the organisers of the Trimester Program ``Prospects of Formal Mathematics'' held in Summer 2024 at the Hausdorff Research Institute for Mathematics in Bonn.
That event provided ring-fenced time, the right atmosphere and invaluable interaction with colleagues, all key factors in stimulating the ideas from which the present research originated.
Funded by the Deutsche Forschungsgemeinschaft (DFG, German Research Foundation) under Germany's Excellence Strategy -- EXC-2047/1 -- 390685813.
The author would like to thank Lancaster University for providing financial support to attend this conference.

\bibliographystyle{plain}
\bibliography{mbc}

\end{document}